\documentclass{xarticle}

\usepackage{graphicx}
\usepackage{overpic}
\usepackage{stmaryrd}
\usepackage{trimclip}

\newcommand{\eemph}[1]{\textbf{\textit{#1}}}
\newcommand{\Z}{\mathbb{Z}}
\newcommand{\R}{\mathbb{R}}
\newcommand{\bmat}[1]{\begin{bmatrix}#1\end{bmatrix}}
\def\one{\mathbf{1}}
\def\tp{\mathsf{T}}
\let\mathopfont=\mathrm
\newcommand{\Null}{\mathop{\mathopfont{null}}}
\newcommand{\rank}{\mathop{\mathopfont{rank}}}
\let\bl\bigl
\let\br\bigr
\newbox\vcbox
\def\vcent#1{\setbox\vcbox\hbox{#1}\raise -0.5\ht\vcbox\hbox{#1}}
\newcommand{\ie}{{\it i.e.}}
\newcommand{\floor}[1]{\lfloor#1\rfloor}
\def\ugn{\lambda}
\newcommand{\bittide}{{bittide\ }}
\def\trec{t_\textup{rec}}
\def\tsend{t_\textup{send}}
\DeclareMathOperator{\receive}{receive}
\DeclareMathOperator{\send}{send}
\def\nodes{\mathcal{V}}
\def\edges{\mathcal{E}}
\def\vext{\mathcal{V}_\text{ext}}
\def\gext{\mathcal{G}_\text{ext}}
\def\eext{\mathcal{E}_\text{ext}}
\def\hateext{\hat{\mathcal{E}}_\text{ext}}
\def\tree{\mathcal{T}}
\def\wu{\omega^\text{u}}
\newcommand{\eg}{{\it e.g.}}
\def\itoj{{i \shortto j}}
\def\jtoi{{j \shortto i}}
\def\ptoq{{p \shortto q}}

\def\jtok{{j \shortto k}}

\makeatletter
\DeclareRobustCommand{\shortto}{\mathrel{\mathpalette\short@to\relax}}
\newcommand{\short@to}[2]{\mkern2mu
  \clipbox{{.5\width} 0 0 0}{$\m@th#1\vphantom{+}{\shortrightarrow}$}}
\makeatother

\begin{document}

\title{Logical Synchrony and the \bittide Mechanism}

\author{%
  Sanjay Lall\footnotesymbol{1}
  \and  C\u{a}lin Ca\c{s}caval\footnotesymbol{2}
  \and  Martin Izzard\footnotesymbol{2}
  \and  Tammo Spalink\footnotesymbol{2}}

\note{Preprint}

\maketitle

\makefootnote{1}{S. Lall is with the Department of Electrical
  Engineering at Stanford University, Stanford, CA 94305, USA, and is
  a Visiting Researcher at Google.\par
  \texttt{lall@stanford.edu}\medskip}

\makefootnote{2}{C\u{a}lin Ca\c{s}caval,
   Martin Izzard, and Tammo Spalink are with Google.}

\begin{abstract}
  We introduce logical synchrony, a framework that allows
  distributed computing to be coordinated as tightly as in synchronous
  systems without the distribution of a global clock or any reference
  to universal time. We develop a model of events called a logical
  synchrony network, in which nodes correspond to processors and every
  node has an associated local clock which generates the events. We
  construct a measure of logical latency and develop its properties.
  A further model, called a multiclock network, is then analyzed and
  shown to be a refinement of the logical synchrony network. We
  present the bittide mechanism as an instantiation of multiclock
  networks, and discuss the clock control mechanism that ensures that
  buffers do not overflow or underflow. Finally we give conditions
  under which a logical synchrony network has an equivalent synchronous
  realization.
\end{abstract}

\section{Introduction}

Distributed computation requires processes on networked
machines to coordinate, presenting challenges in maintaining a
consistent notion of time across nodes.  Local clocks require
continual realignment to prevent divergence, while distributing a
global clock is fragile and expensive at scale. When coordination is
focused on correctness, instead of tracking time an option is to track
only causality.  This takes the form of event sequence information,
such as vector clocks, which avoid needing to synchronize clocks but
remain expensive at scale.

In this paper, we introduce \eemph{logical synchrony}, a novel
approach providing a shared notion of time sufficient for reasoning
about causality without requiring a shared system-wide clock. Logical
synchrony scheduling relies solely on knowledge of graph
topology and logical latencies. We present the \bittide mechanism,
which facilitates efficient implementation of logical synchrony on
modern networks, establishing \emph{synchrony} alongside wall-clock
time as a primary abstraction. By ensuring that clocks advance in
lockstep with data frames sent between nodes, \bittide creates a clock
mechanism with reduced state-keeping compared to vector clocks,
enhancing scalability.

Modern networks, including recent versions of Ethernet, 
continually transmit frames regardless of nodes sending 
\emph{actual data} or not, in
order to maintain synchronization of SerDes~\cite{stauffer}
and clock recovery~\cite{ling} circuits.
Leveraging this, the bittide mechanism achieves logical synchrony
by directly tying the clock advancement to the continuous frame transmission
of such networks.
It is this continual transmission that enables 
\bittide synchronization to occur without the overhead of sending 
additional information, a benefit over explicit synchronization 
protocols such as PTP~\cite{ptp,li_sundial_2020}. Applications
  on networks with clocks 
synchronized to wall-clock time must utilize clock error-bounds for 
correctness reasoning~\cite{corbett_spanner_2013}.  The \bittide
system enables cycle-accurate coordination 
without additional clock error-bounds or any associated barriers. 
This is achieved by defining the clock ordering at neighboring nodes
using a graph of frame transmission events.

Logical synchrony is particularly useful for applications
with predictable behavior and resource requirements, including financial
exchanges~\cite{gupta},
databases~\cite{cockroachdb,li_sundial_2020,corbett_spanner_2013},
robotics~\cite{bateni}, and large-scale numerical computations such as
machine-learning training and inference~\cite{cowan}. Such
predictability allows for 
ahead-of-time scheduling across both communications and computation, 
which in turn allows for high efficiency and bounded response times.  An example 
use case is ensuring concurrency control correctness in lock-free database 
transactions by ensuring that all distributed system nodes observe the
same order of events.

Ahead-of-time scheduling, inherent to logical synchrony, is naturally
limited to applications with predictable communication, memory, and compute cycles.  
Traditional dynamic communication stacks and scheduling infrastructure can 
be implemented above \bittide transparently, which allows running applications which
do not have the requisite predictability, 
but applications running on these 
stacks lose the benefits of ahead-of-time scheduling.  Further research may extend 
the utility of logical synchrony to more dynamic and data-dependent situations, 
for example to support probabilistic ahead-of-time scheduling of such applications 
where behavior is evolving slowly enough for a scheduler to adapt and reconfigure.

Logical synchrony and \bittide have nodes track logical time, 
which potentially diverges from wall-clock time. This poses a limitation for 
applications requiring wall-clock time, such as real-time embedded systems or 
control systems. 
Addressing this limitation and failure handling requires
augmentation of the basic \bittide mechanism presented here, and thus are beyond 
the scope of this paper.  A consequence of ahead-of-time scheduling is that failure
handling naturally happens independently of scheduling because there is no
runtime dynamic scheduler.  Node or link failures may necessitate rescheduling
execution or communication, which in turn may require application participation.

\hyphenation{dis-tri-buted}

The bittide mechanism enables processes on distributed network 
cores to behave as if perfectly synchronized despite individual cores 
being only imperfectly synchronized. A logical synchrony network, abstracting 
the \bittide mechanism, characterizes causality relationships between events. 
Logical latencies specify these relationships exactly, a striking property 
that allows precise coordination and reasoning about both system performance 
and event ordering. Using communication events between processes for logical 
coordination originates with the work of Lamport~\cite{lamport} and allows 
precise reasoning about correctness.  Logical synchrony ties these events to 
the repetitive frame transmission events of the network, and thus allows 
precise coordination and reasoning about the performance of the system as well
as the ordering of events, bringing the guarantees available in
synchronous execution to distributed systems without the need for a
global time reference.  Our work extends Lamport's framework 
into the efficiency domain, enabling reasoning about both correctness and 
scheduling.

Synchronous execution models have been used successfully in realtime
systems~\cite{berry, benveniste:12yearsofsynchrony:2003,
roop_towards_2004} to reason about correctness, in particular meeting
deadlines.  Often, synchronous abstractions are decoupled from
implementation and are used to validate system functional behavior.
When mapping synchronous abstractions to asynchronous
non-deterministic hardware, work has been done to automate code
generation that matches the functional semantics, hiding the
non-deterministic behavior of the hardware with explicit
synchronization, for example~\cite{didier_sheep_2019}.  Logical
Execution Time (LET) was introduced by Henzinger and
Kirsch~\cite{henzinger_let:2003} to support the design of reactive,
cyber-physical systems. More recently, Lingua
Franca~\cite{linguafranca,linguafranca:time:lncs:2021} supports
concurrent and distributed programming using time-stamped
messages. Lingua Franca exposes to programmers the notion of {\em
reactors} that are triggered in logical time, allowing deterministic
reasoning about four common design patterns in distributed systems:
alignment, precedence, simultaneity, and consistency. We argue that
the causality reasoning in the logical synchrony framework subsumes
such design patterns -- they are all effectively enabling reasoning
about ordering of events in a system that exchanges messages, and as
we will show in the paper, this is exactly the class of applications
for which logical synchrony determines precisely the causality
relationships.

Alternatively, synchronous execution can be implemented using a
single global clock. For small real-time systems, cyber-physical
systems, and control systems, a global clock can be distributed from a
single oscillator. Scaling such systems is difficult because large clock
distribution networks introduce delays which must be corrected.  

Preceding works
  such as Sundial~\cite{li_sundial_2020} have also showcased the
  difficulty in managing fault tolerance for synchronized real-time
  clocks.
For the majority of systems using wall-clock time as their global
clock, synchronization implies exchanging timestamps~\cite{ptp, ntp}.
Techniques such as TrueTime~\cite{corbett_spanner_2013} and
White Rabbit~\cite{whiterabbit} attempt to reduce the latency
uncertainty, and thus the time-uncertainty bounds, from milliseconds in
TrueTime to sub-nanosecond in White Rabbit.

To achieve desired levels of
performance using existing network protocols requires expensive time
references such as dedicated atomic clocks and networking hardware
enhancements to reduce protocol overhead. Time uncertainty is exposed
to programmers through an uncertainty interval which guarantees that
current time is within interval bounds for all nodes in the system, such
that every node is guaranteed to have passed current time when the bound
elapses.

Logical synchrony, formalized in Section~\ref{sec:logical-synchrony},
abstracts the notion of shared time and allows us to avoid a global
reference clock or wall-clock. Time is defined only by local clocks
decoupled from physical time. The idea is that events at the same node
are ordered by local time, and events at different nodes are ordered
by causality.  As we will show, logical synchrony requires no
system-wide global clock and no explicit synchronization (timestamp
exchanges or similar), which thereby allows for potentially infinitely
scalable systems. Reasoning about ordering of events in logically
synchronous systems follows the partial order semantics of
Lamport~\cite{lamport} and thus provides equivalence with any
synchronous execution that generates identical event graphs.

To establish how logical synchrony can be realized in practice, we
first define what logical synchrony means within an abstract model of
distributed systems with multiple clocks, defining local clocks in a
multiclock network. We show how to combine the FIFO occupancies
with the offsets between neighboring clocks, and how this combination 
is enough to determine the causality relationships.

We then explain how \bittide~\cite{bms,res,reframing} is a mechanism to 
efficiently implement logical synchrony with real hardware and thereby 
bring desirable synchronous execution properties to distributed 
applications efficiently at scale.

\subsection{Mathematical preliminaries and notation}

An \eemph{undirected graph} $\mathcal G$ is pair $(\mathcal V,
\mathcal E)$ where $\mathcal V$ is a set and $\mathcal E$ is a subset
of the set of 2-element subsets of $\mathcal V$.
A \eemph{directed graph} $\mathcal G$ is pair $(\mathcal V, \mathcal
E)$ where $\mathcal E \subset \mathcal V \times \mathcal V$ and $(v,
v)\not\in \mathcal E$ for all $v \in\mathcal V$.  An edge
$e\in\mathcal E$ in a directed graph may be denoted $(u, v)$ or $u \to
v$.  A directed graph may contain a 2-cycle, that is a pair of edges
$u \to v$ and $v \to u$.  An \eemph{oriented graph} is a directed
graph in which there are no 2-cycles.

Suppose $G = (\mathcal V, \mathcal E)$ is a directed graph, and number
the vertices and edges so that $\mathcal V = \{1, \dots,n\}$ and
$\mathcal E = \{ 1, \dots, m\}$.  Then the \eemph{incidence matrix} $B
\in\R^{n\times m}$ is
\[
B_{ij} = \begin{cases}
  1 & \text{if edge $j$ starts at node $i$} \\
  -1 & \text{if edge $j$ ends at node $i$} \\
  0 & \text{otherwise}
\end{cases}
\]
for $i=1,\dots,n$ and $j=1,\dots,m$.

A \eemph{walk} in a directed graph $G$ is a non-empty alternating
sequence $v_0,s_0,v_1,s_1,\dots,s_{k-1},v_k$ in which $v_i \in
\mathcal V$, $s_i \in \mathcal E$, and either $s_i = v_i \to v_{i+1}$
or $s_i = v_{i+1} \to v_i$.  In the former case we say $s_i$ has
\emph{forward} or $+1$ orientation, otherwise we say it has
\emph{backward} or $-1$ orientation.  A \eemph{path} is a walk in
which all vertices are distinct.  A \eemph{cycle} is a walk in which
vertices $v_0,\dots,v_{k-1}$ are distinct, all edges are distinct, and
$v_0 = v_k$.  Walks, paths, and cycles are called \eemph{directed} if
all edges are in the forward orientation.

In a directed
graph $G$, given a walk
\[
W =
(v_0,s_0,v_1,s_1,\dots,s_{k-1},v_k)
\]
the corresponding \eemph{incidence vector} $x\in\R^m$ is such that
$x_i = 1$ if there exists $j$ such that $i = s_j$ and $s_j$ has
forward orientation, and $x_i=-1$ if there exists $j$ such that $i =
s_j$ and $s_j$ has reverse orientation, and $x_i=0$ otherwise. For a
directed graph with 2-cycles, there is an edge $u \to v$ and $v \to
u$, and we assign one of these directions as primary and the other as
secondary. This is simply a choice of sign convention.  From a
directed graph we construct an associated oriented graph by discarding
all secondary edges. From an oriented graph we construct an associated
undirected graph by discarding all orientations.  The concepts of
spanning tree and connectedness when applied to a directed graph
always refer to the associated undirected graph. The following two
results are well-known.

\begin{thm}
  \label{thm:smith}
  Suppose $\mathcal G = (\mathcal V, \mathcal E)$ is a directed graph
  with incidence matrix $B$, and suppose edges $1,\dots,n-1$ form a
  spanning tree.
  Partition $B$ according to
  \[
  B = \bmat{B_{11}  & B_{12} \\ -\one^\tp B_{11} & -\one^\tp B_{12}}
  \]
  then $B_{11}$ is unimodular. Further
  \[
  B =
  \bmat{ B_{11} & 0 \\ -\one^\tp B_{11} & 1}
  \bmat{I & 0 \\ 0 & 0}
  \bmat{I & N \\ 0 & I}
  \]
  where $N = B_{11}^{-1}B_{12}$.
\end{thm}
\begin{proof}
  See for example Theorem~2.10 of~\cite{bapat}.
\end{proof}

For convenience, denote by $Z$ the $m \times (m-n+1)$ matrix
\[
Z = \bmat{-N \\ I}
\]
Then we have the following important property.
\begin{thm}
  \label{thm:incvec}
  Every column of $Z$ is the incidence vector of a cycle in $\mathcal{G}$.
\end{thm}

\begin{proof}
  See, for example, Chapter 5 of~\cite{bapat}.
\end{proof}

Theorem~\ref{thm:smith} implies that the columns of $Z$ are a basis
for the null space of $B$, since $BZ=0$ and $\Null(Z) = \{0\}$. The columns
of $Z$ are called the
\eemph{fundamental cycles} of the graph. Note that each of the
fundamental cycles is associated with exactly one of the non-tree
edges of the graph.

\section{Logical synchrony networks}
\label{sec:logical-synchrony}

The goal of this section is to develop an abstraction which contains
two key things; first, a notion of ordering of events such as that of
Lamport~\cite{lamport}; and second, a notion of network latency.  It
turns out that these two ideas may be combined into a simple unified
abstraction, which we call the \emph{logical synchrony network}, and
this allows analysis of both causality and system
performance. We build an event model, in which events may be thought
of as ticks of a local clock at each node, corresponding to process
execution. The events at neighboring nodes are linked by data
transmission. There is no notion of global time, and yet
within this framework there is still a notion of latency and duration.
We show that ordering of events can be defined in a meaningful way
when round-trip latencies are positive.

We start with a formal definition of a logical synchrony network as a
directed graph with edge weights, as follows.

\begin{defn}
  A \eemph{logical synchrony network} is a directed graph
  $(\nodes,  \edges)$ together with a set of edge
  weights $\ugn : \edges  \rightarrow \Z$.
\end{defn}

In this model, each node corresponds to a processor, and an edge
between nodes $i \to j$ indicates that node~$i$ can send data along a
physical link to node~$j$.  Sent data is divided into tokens which we
refer to as \emph{frames}.

\paragraph{Local clocks.}
Every node has an infinite sequence of \eemph{events} associated with
it, which can be thought of as compute steps. The events at node $i$
are denoted $(i,\tau)$, where $\tau$ is referred to as a
\emph{localtick} and thereby implicitly defines a local clock.  We
define the set of all events
\[
\vext = \{ (i,\tau) \mid i\in\nodes, \tau \in \Z\}
\]
Events at one node are aligned to events at other nodes by the
transmission of frames. At localtick $\tau$ and node $i$, a frame is
sent from node $i$ to node $j$, and it arrives at node $j$ at
localtick $\tau + \lambda_\itoj$. The constant $\ugn_\itoj$ is called
the \emph{logical latency.} We define the following binary relation.

\begin{defn}
  \label{def:eg}
  Event $(i,\tau)$ is said to
  \eemph{directly send to} the event $(j,\rho)$ if $(i,j) \in \edges$
  and $\rho = \tau + \ugn_\itoj$, or $i =j $ and $\rho = \tau + 1$. We
  use the notation
  \[
  (i,\tau) \to (j,\rho)
  \]
  to mean $(i,\tau)$ directly sends to $(j,\rho)$, and define the set
  \[
  \eext = \{ \bl((i,\tau), (j,\rho)\br) \mid
    (i,\tau) \to (j,\rho)\}
  \]
  The graph $\gext = (\vext, \eext)$
  is called the \eemph{extended graph} of the logical synchrony network.
\end{defn}
This relation may be viewed as an infinite directed graph with vertex
set $\vext$ and directed edges $(i,\tau) \to (j,\rho)$.  In this
graph, those edges $(i,\tau) \to (j,\rho)$ for which $i=j$ are called
\emph{computational edges}. An edge that is not a computational edge
is called a \emph{communication edge}.  Figure~\ref{fig:extgraph}
illustrates a logical synchrony network and its corresponding extended
graph. Definition~\ref{def:eg} adds two types
  of edges to the extended graph. Computational edges are vertical
  in the figure, and they connect $(i,\tau)$ to $(i,\tau+1)$.
  These express the relationship between sequential events at node $i$.
      Communication edges are non-vertical, and connect $(i,\tau)$ to
      $(j,\tau + \ugn_\itoj)$. These express the relationship between
      the sending of a frame from node $i$ at time $\tau$ and its
      reception at node $j$ at time $\tau + \ugn_\itoj$.

\begin{figure}[htb]
  \hbox to \linewidth{\hss
    \vcent{\hbox{\includegraphics[width = 30mm]{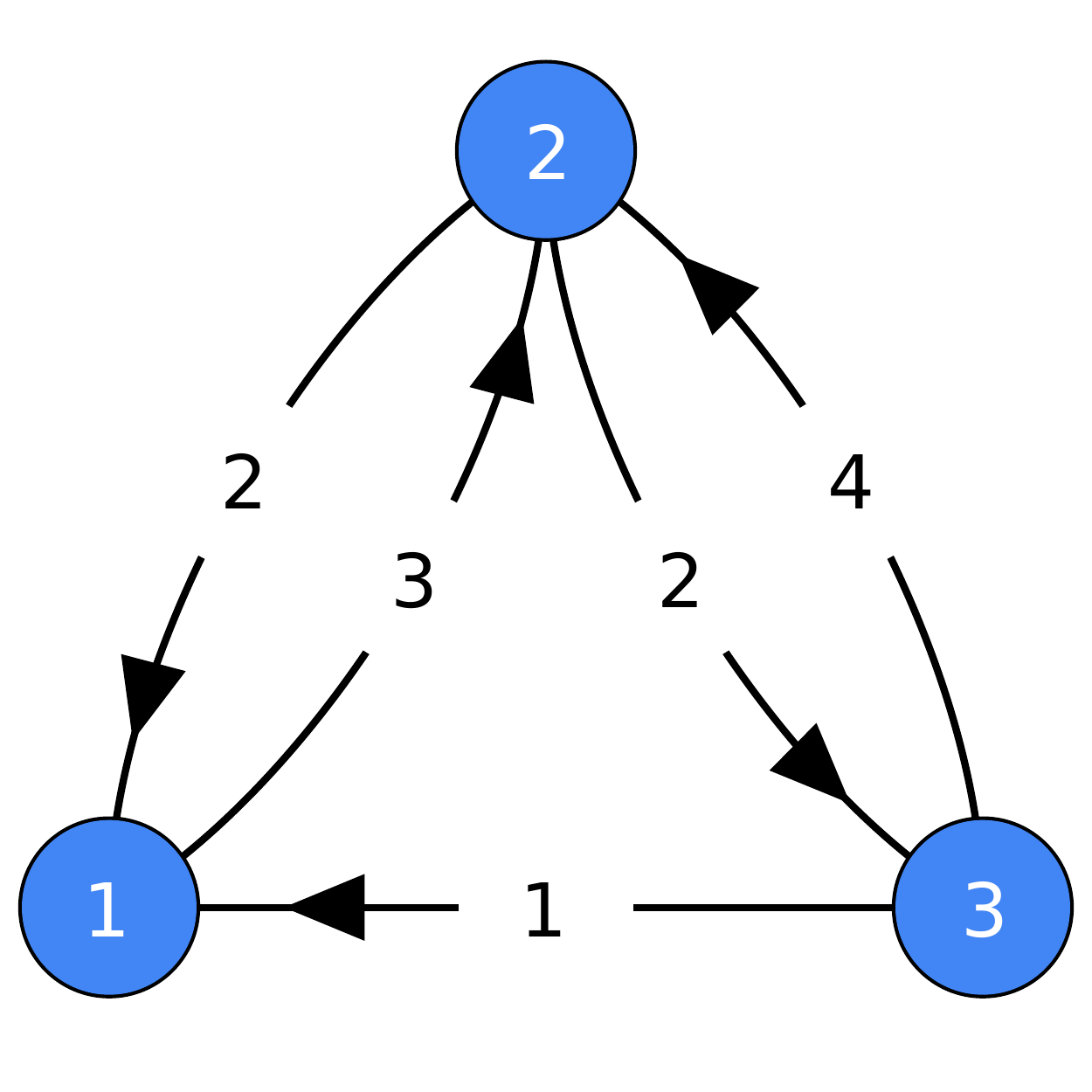}}}\hss
    \vcent{\hbox{\includegraphics[width = 40mm]{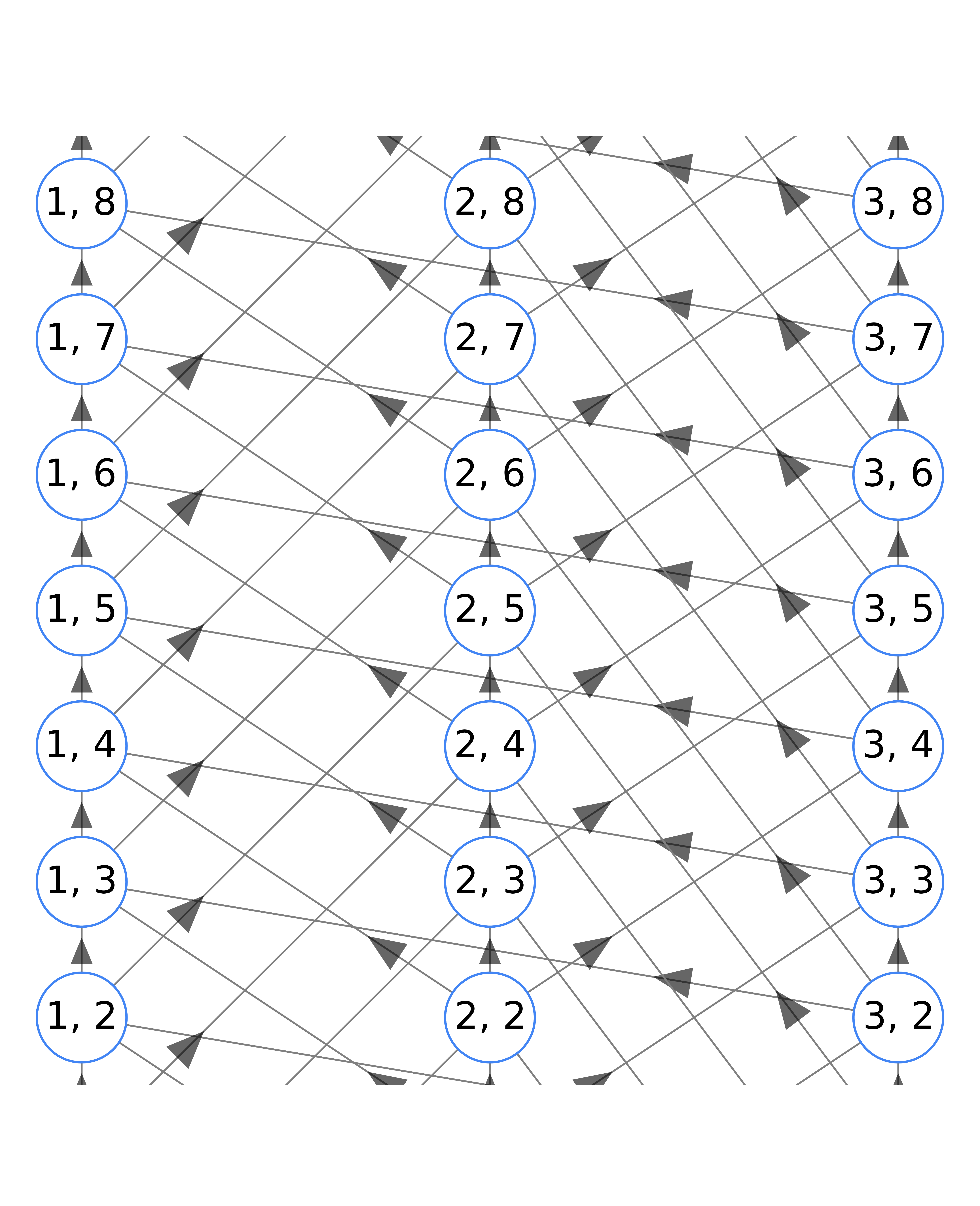}}}\hss}
  \caption{A logical synchrony network (edges labeled with $\lambda$)
  and corresponding extended graph.}
  \label{fig:extgraph}
\end{figure}

The localticks define a separate and ideal notion of local duration
at each node by counting events (\ie, frame transmissions or receptions.)
We can speak of the event $(i,\tau)$ as occurring at time $\tau$
localticks on node $i$.  We say that event $(i,\tau + a)$ happens $a$
localticks after event $(i,\tau)$, for any $a\in\Z$.  We cannot in general
compare clock values at two different nodes.

\paragraph{Execution.}
This model captures the local evolution of time at each node
$i\in\nodes$, and the transmission of frames between them. Although we
do not investigate execution models in this paper, it is possible to
define many different execution semantics. One simple choice is the
functional model, where frames carry data, and associated with each
event $(i,\tau)\in\vext$ in the extended graph we have a function,
which maps data from incoming edges to data on outgoing edges. Another
possibility is to have a more procedural model, where events in
$\vext$ correspond to the clock ticks of a processor in the
corresponding $\nodes$.  For the purposes of this paper it is not
necessary to specify how many bits each frame contains but we assume
all frames on a given link are equally sized.

The abstract models considered in this paper consist of sequences of
events which extend infinitely far into both the future and the past.
It is possible to extend this model to include system startup, for
example by introducing a minimum node within the extended graph, or by
modifying the execution model.  We do not address startup
within this paper.

\paragraph{Frames and logical latency.}

If $A$ denotes a particular frame sent
$i\to j$, then we will make use of the notation $\receive(A)$ to refer
to the localtick at node $j$ when $A$ arrives at $j$.  Similarly
$\send(A)$ refers to the localtick at node $i$ when $A$ was sent.
This notation leaves implicit the source and destination of frame $A$,
in that $i,j$ are not included as arguments of the $\send$ and
$\receive$ functions. We do not as yet assume any particular mechanism
for transmission of frames, but we assume that frames are received in
the order that they are sent, without any loss. Note that the logical
latency has no connection to \emph{physical latency}. If we were to
measure the send and receive times with respect to a global notion of
time, we would know that, for example, the receive time must be
greater than the send time.  In the framework presented here, that is
not the case; the localticks are strictly local, and as a result there
is no such requirement on their numerical value; the logical latency
$\lambda_\itoj$ may be negative. This is, of course, a statement about
the clocks, not about causality.

In words, the logical latency is the time of arrival \emph{in the
receiver's clock} minus the time of departure \emph{in the sender's
clock}.  There are several observations worth making about logical
latency.

\begin{itemize}

\item Logical latency is \emph{constant}. For any two nodes $i, j$,
  every frame sent $i \to j$ has the same logical latency.  It is a
  property of the edge $i \to j$ in $\edges$.

\item Despite the name, logical latency is not a measure of length
  of time or duration. It is not the case that if $\ugn_\itoj$ is
  greater than $\ugn_\ptoq$ then it takes longer for frames to
  move from $i$ to $j$ than it does for frames to move from $p$ to
  $q$. (In fact, we do not have a way within this framework to compare
  two such quantities.)

\item The logical latency can be negative.

\end{itemize}

\paragraph{Logical latencies and paths.}
Logical latencies add along a path. Suppose node $i$ sends a frame $B$
along edge $i \to j$ to node $j$, and then node $j$ forwards it
$j \to k$.  Then we have
\[
\receive(B) =  \send(B) + \ugn_\itoj + \ugn_\jtok
\]
This means that we can speak of the logical latency of the path $i \to
j \to k$ as being $\ugn_\itoj +\ugn_\jtok$, and more generally we can
define the logical latency of a directed path
$\mathcal P=v_0,s_0,v_1,s_1,\dots,s_{k-1},v_k$ from node $v_0$ to
node $v_k$ in $\mathcal G$.
The logical latency is
path dependent; two paths with the same endpoints may have different
logical latencies.  We have
\[
\lambda_{\mathcal P} = \sum_{i=0}^{k-1} \lambda_{s_i}
\]
This makes sense, which is potentially surprising because we are
measuring arrival and departure times with different clocks.  Since
frames are being relayed, there may be additional delay at
intermediate nodes (\ie, additional compute steps)
which would need to be included when determining the destination event.
Logical latencies are defined such that they do not included
this additional delay.

\subsection{Ordering of events}

A fundamental question regarding causality arises in the study of
distributed systems.  Given two events, we would like to determine
which happened first. In a nonrelativistic physical setting, such a
question is well-defined. In a relativistic setting, there are events
which are separated in space for which the relative order is
undetermined --- the order depends on the observer. Something similar
happens in distributed systems, as was pointed out by
Lamport~\cite{lamport}.  Given two events, instead of asking which
event happened first, a more useful question is to ask which event, if
any, \emph{must have} happened first.  The framework for distributed
clocks developed by Lamport~\cite{lamport} established that there is a
partial ordering on events determined by one event's ability to
influence another by the sending of messages. In that paper the author
defines a global notion of time consistent with said partial
order. Subsequent work~\cite{fidge,mattern} defines \emph{vector
clocks} which assign a vector-valued time to events for which the
partial ordering is equivalent to that defined by message-passing.  We
would like to construct the corresponding notion of causality in a
logical synchrony network.

We define below the $\sqsubset$ relation, which can be used to define
a partial order on $\gext$  provided we can ensure that it is
acyclic. To do this, we consider round-trip times.

\paragraph{Round trip times.}
Logical latencies are not physical latencies, despite the additive
property. However, there is one special case where logical latency is
readily interpreted in such physical terms, specifically the time for
a frame $A$ to traverse a cycle in the graph, the cycle round-trip
time.  Suppose $\mathcal C = v_0,s_0,v_1,s_1,\dots,s_{k-1},v_k$ is a
directed cycle, then
\[
\lambda_\mathcal{C} = \receive(A) - \send(A)
\]
is the round-trip time measured in localticks. Two different cycles
from a single node $i$ may have different round-trip times, and these
are comparable durations since they are both measured in localticks at
that node. We have
\[
\lambda_\mathcal{C} = \sum_{i=0}^{k-1} \ugn_{s_i}
\]
We make the following definition.
\begin{defn}
  \label{defn:lsacyclic}
  A logical synchrony network is said to have \eemph{positive
    round-trip times} if, for every directed cycle $\mathcal C $ in
  the graph $\mathcal G$ we have $\lambda_\mathcal{C} > 0$.
\end{defn}

We then have the following result, which says that if the round-trip
times around every directed cycle in the logical synchrony network are
positive, then the extended graph is acyclic.
\begin{thm}
  If a logical synchrony network has positive round-trip times then its
  extended graph is acyclic.
\end{thm}
\begin{proof}
  Suppose for the purpose of a contradiction that the extended graph is cyclic. Then
  there exists a directed cycle
  $\mathcal C_1 =  v_0,s_0,v_1,s_1,\dots,s_{k-1},v_k$
  where each $v_j \in \vext$ is a pair $v_j
  = (i_j,\tau_j)$.
  Since the start and end node is the same, we have
  \begin{equation}
    \label{eqn:contrad}
    \begin{aligned}
    0 & = \sum_{j=1}^{k-1} (\tau_{j+1} - \tau_j) \\
    &= \sum_{j \in C_\text{comp}}
      (\tau_{j+1} - \tau_j)
    + \sum_{j \not\in C_\text{comp}} (\tau_{j+1} - \tau_j)
    \end{aligned}
  \end{equation}
  where $C_\text{comp}$ is the set of indices $j$ such that
  $(v_j,v_{j+1})$ is a computational edge.  Each of the computational
  edges has $\tau_{j+1} - \tau_j = 1$.  If all of the edges in the
  graph are computational then the right-hand side is positive.  If
  there are some communication edges, then the second of the two terms
  on the right-hand side is positive due to the assumption that the
  logical synchrony graph has positive round-trip times, and again the
  right-hand-side is positive. This contradicts the claim that the sum
  is zero.
\end{proof}

This acyclic property is necessary for an execution model based on
function composition to be well-defined.  It also allows us to define
a temporal partial ordering between events in $\gext$. Since a logical
synchrony network with positive round-trip times has an extended graph
which is acyclic, the reachability relation on the extended graph
defines a partial order.  Specifically, we write
\[
 (i,\tau) \sqsubset (j,\rho)
\]
if there is a directed path from $(i,\tau)$ to $(j,\rho)$ in the
extended graph.  Here, the notation is meant to be similar to $<$,
indicating \emph{comes before}. Under these conditions, a logical
synchrony network is a distributed system in the sense of
Lamport~\cite{lamport}, with logical latencies providing strict
inter-event timings at any node $i\in\nodes$.
The partial ordering on the induced logical synchrony network has
exactly the property that, if $u \sqsubset v$, then $u$ must have
happened before $v$.

\section{Equivalence of LSNs}

The goal of this section is to establish an invariant,
which we will use in the subsequent sections to analyze system
correctness.  We introduce the idea of clock relabeling, which
modifies logical latencies while preserving the interconnection of
events and the underlying physical system. We show that the round trip
times, being physically measurable properties, cannot change. We use
this invariant to characterize when two networks are physically the
same, even though their clock labels may be different.

Two logical synchrony networks may have different logical latencies,
but be nonetheless equivalent for the purpose of executing
processes. An example is given by the graphs in
Figure~\ref{fig:triequiv}.

\begin{figure}[htb]
  \hbox to \linewidth{\hss
    \includegraphics[width=30mm]{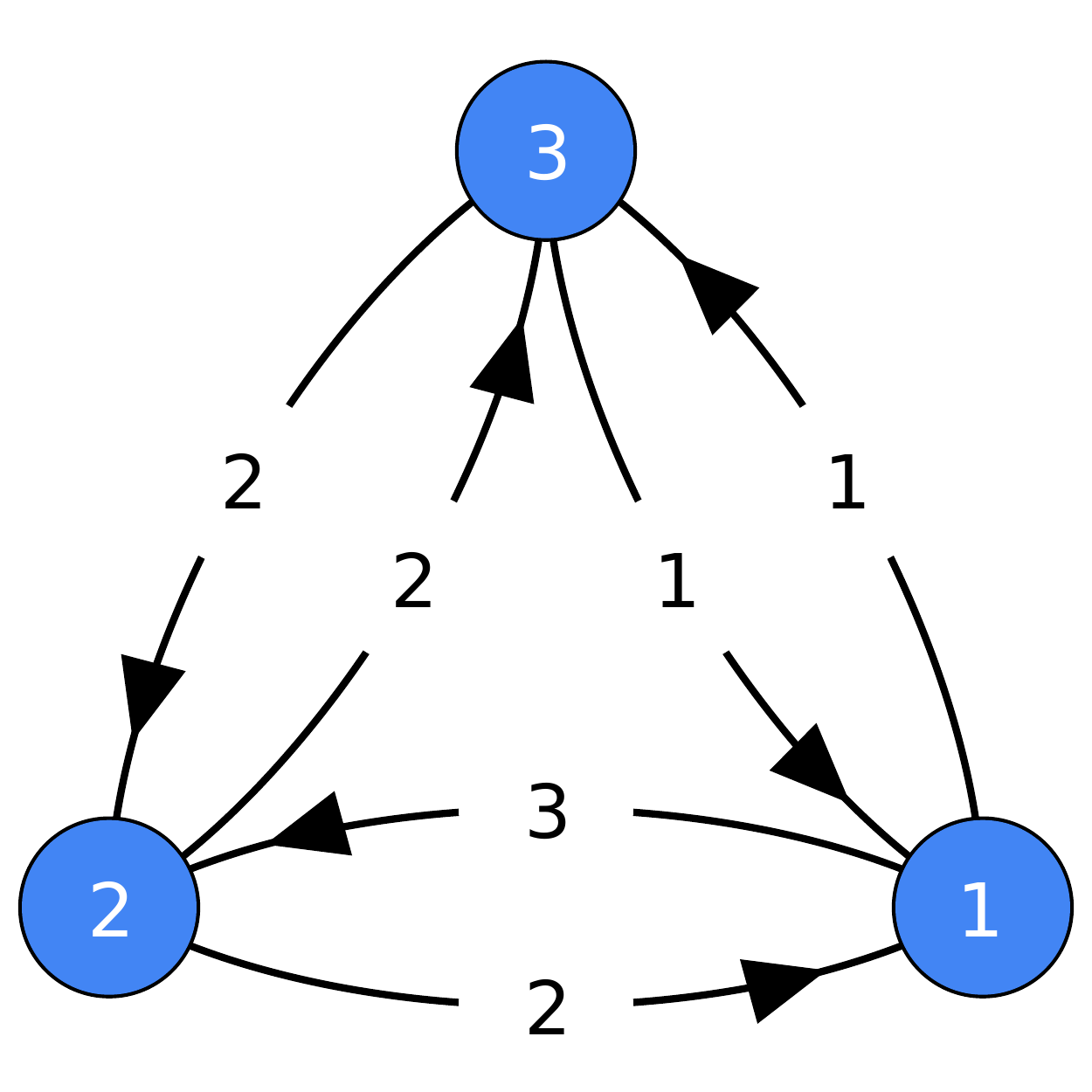}\hss
    \includegraphics[width=30mm]{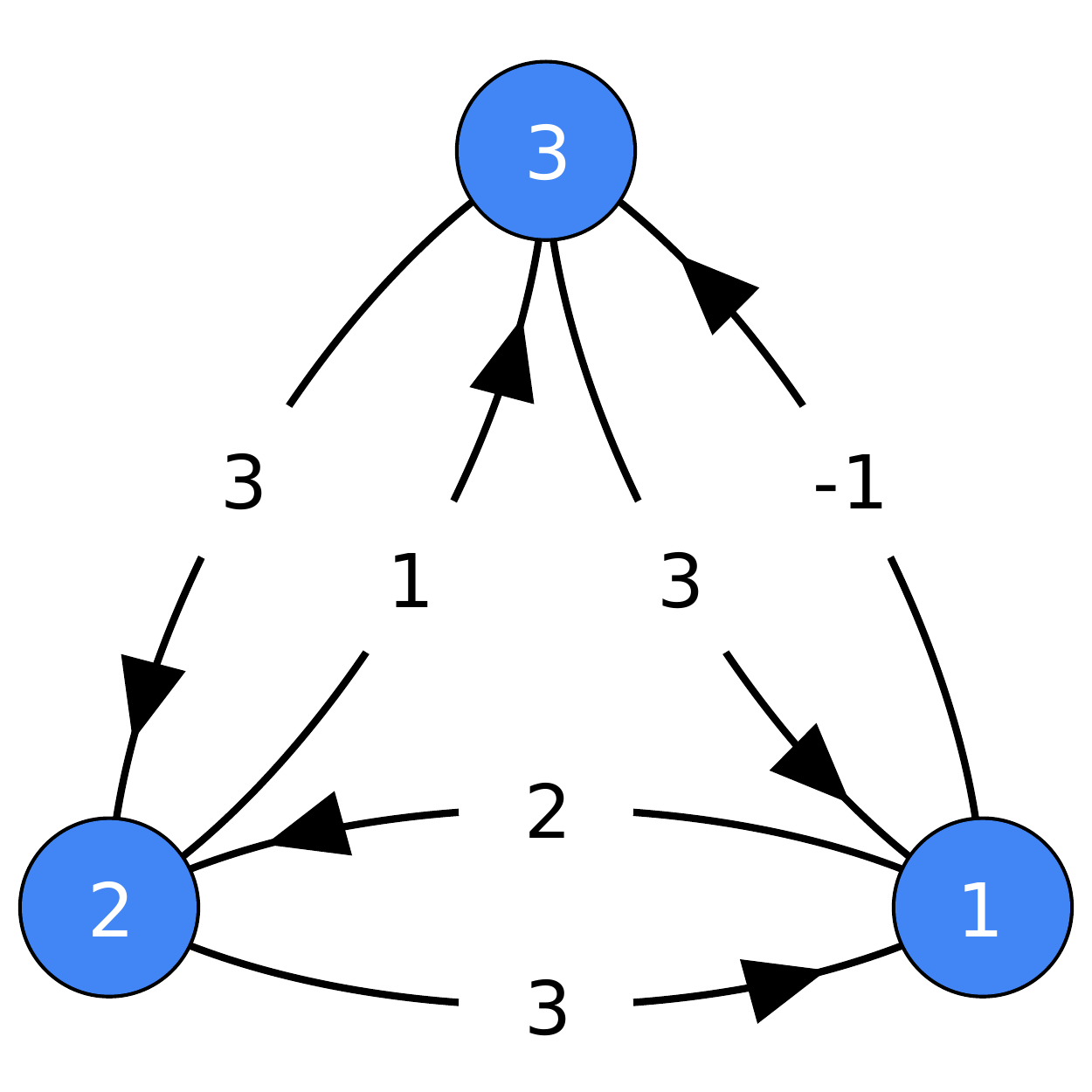}\hss}
  \caption{Two equivalent logical synchrony graphs (edges labeled
    with~$\lambda$). Relabeling the clocks using $c=(1,2,3)$ maps the
    left-hand graph to the right-hand one.}
  \label{fig:triequiv}
\end{figure}

This arises because we can relabel the events. Specifically, given a
logical synchrony network with events $\vext$, we define a new logical
synchrony network. Given $c_1,\dots,c_n\in\Z$, we relabel event
$(i,\tau)$ as $(i, \tau + c_i)$. This is a relabeling of the vertices
of the graph $\gext$. In $\gext$ we have edges
\[
(i,\tau) \to (j, \tau + \lambda_\itoj)
\]
for every $i\neq j\in\nodes$ and $\tau \in \Z$. Under the relabeling,
these are mapped to
\[
(i,\tau + c_i) \to (j, \tau + \lambda_\itoj + c_j)
\]
and since there is such an edge for all $\tau \in \Z$ the edge
set of the relabeled extended graph is
\[
\hateext = \{
\bl((i,\tau), (j, \tau + \lambda_\itoj + c_j - c_i)\br)
\mid i,j\in\nodes, \tau \in \Z\}
\]
This is the extended graph for a logical synchrony network
with logical latencies
\[
\hat\ugn_\itoj= \ugn_\itoj + c_j - c_i
\]
This leads us to the following definition of equivalence.
\begin{defn}
  Suppose we have two logical synchrony networks on a directed graph
  $(\nodes,  \edges)$, with edge weights $\ugn$ and $\hat\ugn$. We
  say these LSNs
  are \eemph{equivalent} if there exists $c_1,\dots,c_n \in \Z$ such that,
  for all $i,j\in\mathcal V$,
  \begin{equation}
    \label{eqn:equiv}
    \hat\ugn_\itoj = \ugn_\itoj + c_j - c_i
  \end{equation}
\end{defn}
\noindent We can write this equation as
\[
\ugn - \hat\ugn  = B^\tp c
\]
where $B$ is the incidence matrix of $\mathcal G$.  Relabeling the
clocks results in a relabeling of the corresponding extended
graph. Since this only changes the labels of the nodes, not how the
nodes are interconnected, any code which is executable on one graph
may also be executed on the other (but any references to particular
localticks will need to be changed.)  Physically measurable properties
such as round-trip times cannot change under such a simple
relabeling. We have
\begin{prop}
  \label{prop:same}
If two LSNs are equivalent, they will have the same round
trip times on every directed cycle.
\end{prop}
\begin{proof}
  The round-trip times
  for a directed cycle $\mathcal C = v_0,s_0,v_1,s_1,\dots,s_{k-1},v_k$
  in $\mathcal G$ satisfy
  \[
  \sum_{j=0}^{k-1} \ugn_{s_j}
  =
  \sum_{j=0}^{k-1} \hat\ugn_{s_j}
  \]
  which follows from equation~\eqref{eqn:equiv}.
\end{proof}

The converse is not generally true, as the following example shows.
\begin{figure}[htb]
  \hbox to \linewidth{\hss
    \includegraphics[width=30mm]{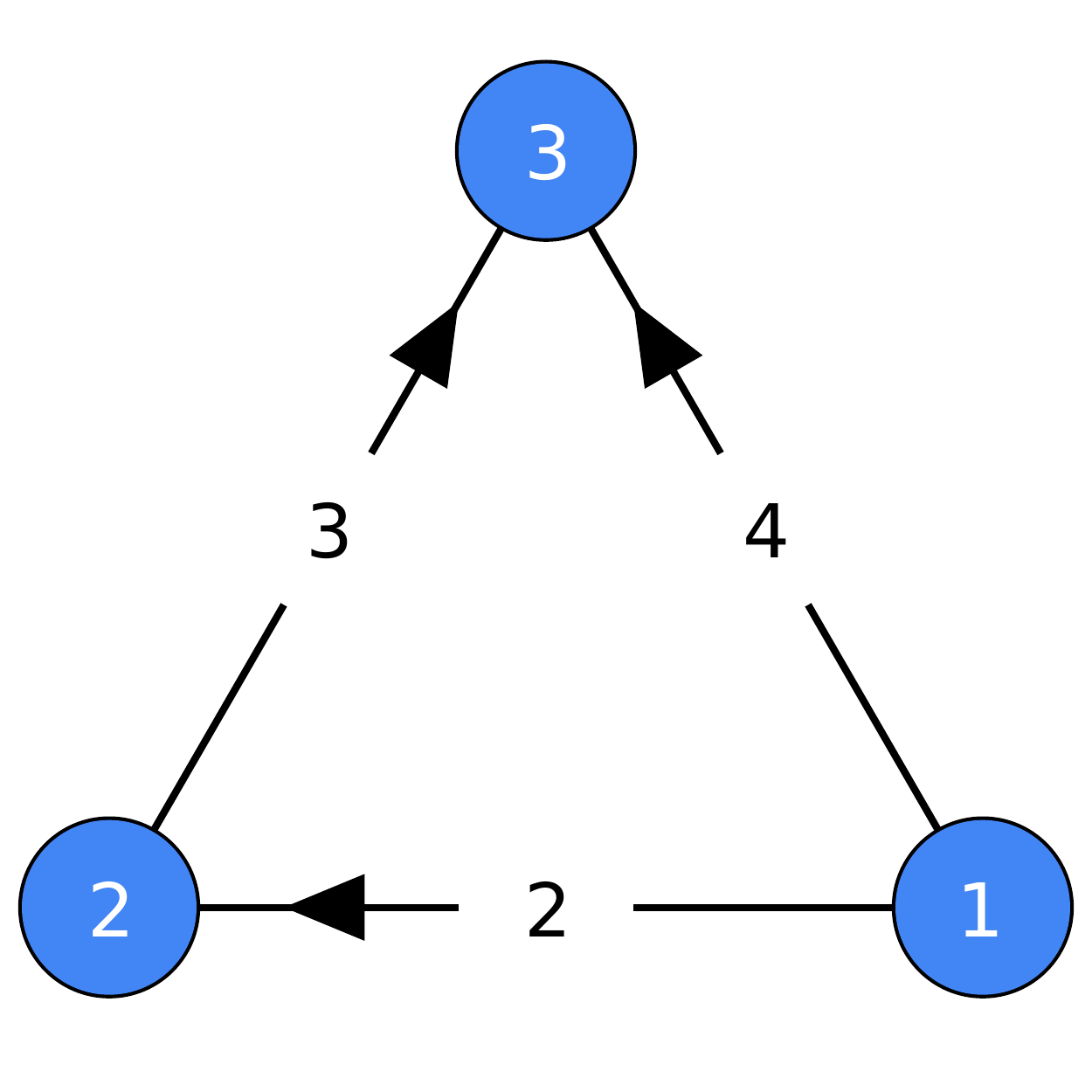}\hss
    \includegraphics[width=30mm]{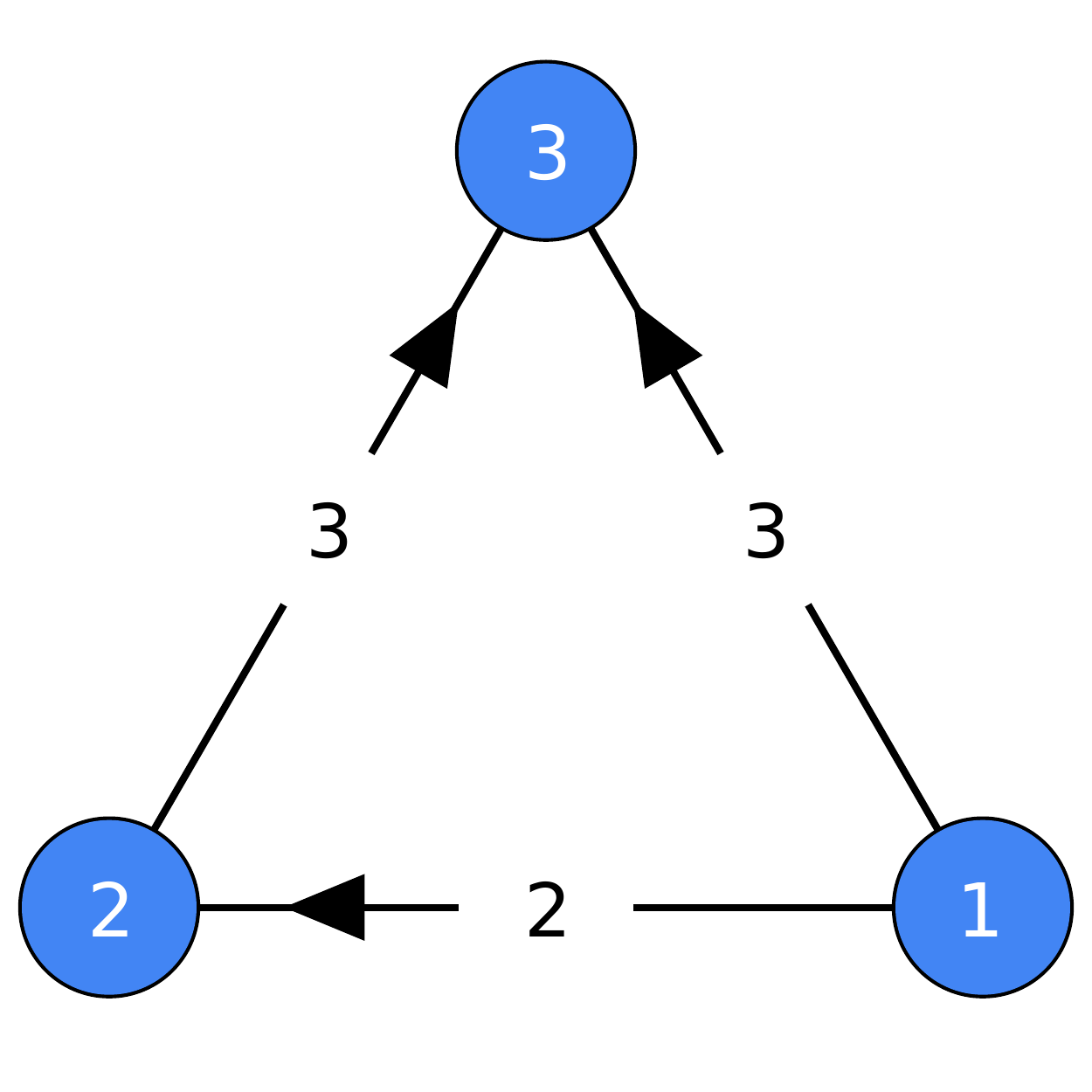}\hss}
  \caption{Two non-equivalent logical synchrony graphs with no
  directed cycles (edges labeled with~$\lambda$)}
  \label{fig:tri1}
\end{figure}

\begin{example}
  \label{ex:rtdirectedcycles}
  Consider the logical synchrony networks shown in
  Figure~\ref{fig:tri1}.  Both networks have the same underlying graph,
  which has no directed cycles, and so the round trip times on every
  directed cycle are trivially equal on both networks.
  If we order the edges $((1\to 2), (2\to 3), (1 \to 3))$ then we have
  incidence matrix
  \[
  B = \bmat{ 1 & 0 & 1 \\
    -1 & 1 & 0 \\
    0 & -1 & -1}
  \]
  which has $\rank(B)=2$. In the left-hand network of
  Figure~\ref{fig:tri1} the logical latencies are $\lambda_1 = 2$,
  $\lambda_2 = 3$ and $\lambda_3 = 4$, and in the right-hand network
  they are $\hat\lambda_1 = 2$ $ \hat\lambda_2 = 3$ and $\hat\lambda_3
  = 3$.
  Therefore
  \begin{equation}
    \label{eqn:counter}
    \ugn - \hat \ugn = \bmat{ 0 \\ 0 \\ 1}
  \end{equation}
  and there is no vector $c$ such that $\ugn - \hat \ugn = B^\tp c$.
\end{example}

If the round trip times are equal around every cycle, accounting for
signs and orientations, then the two logical synchrony networks are
equivalent. To show this, we need a preliminary result.

\begin{lemma}
  \label{lem:intdual}
  Let the graph be connected.
  Suppose $y \in\Z^m$, and for every cycle $\mathcal C$ we have
  $y^\tp x  = 0$ for the corresponding incidence vector $x$. Then
  $y =  B^\tp c $ for some $c\in \Z^n$.
\end{lemma}
\begin{proof}
  Pick a spanning tree, and partition $B$ according to the spanning tree.
  Let $N = B_{11}^{-1} B_{12}$.
  Partition $y$ according to
  \[
  y = \bmat{y_1 \\ y_2}
  \]
  where $y_1\in\Z^{n-1}$.   We choose
  \[
  c = \bmat{B_{11}^{-\tp} y_1 \\ 0}
  \]
  and note that since $B_{11}$ is unimodular  $c$ must be integral.
  Using Theorem~\ref{thm:smith} we have
  \begin{align*}
  B^\tp c
  &= \bmat{I & 0 \\ N^\tp & I} \bmat{I & 0 \\ 0 & 0}
  \bmat{ B_{11}^\tp & -B_{11}^\tp \one \\ 0 & 1}\bmat{B_{11}^{-\tp} y_1 \\ 0}
  \\
  &=  \bmat{I & 0 \\ N^\tp & I} \bmat{y_1 \\ 0} \\
  &= \bmat{y_1 \\ y_2}
  \end{align*}
  as desired,   where in the last line we use Theorem~\ref{thm:incvec} to show that
  \[
  y^\tp\bmat{-N  \\ I}  = 0
  \]
  since $y$ is orthogonal to the incidence vectors of the fundamental cycles.
\end{proof}

We now state and prove a variant of Proposition~\ref{prop:same} which is both
necessary and sufficient.
\begin{thm}
  Suppose we have two logical synchrony networks on a connected directed
  graph $(\nodes, \edges)$, with edge weights $\ugn$ and
  $\hat\ugn$. These networks are equivalent if and only if they have the
  same signed round trip times on every cycle in $\mathcal G$. That is,
  for every cycle $\mathcal C=  v_0,s_0,v_1,s_1,\dots,s_{k-1},v_k$ we have
  \begin{equation}
    \label{eqn:signedcycles}
    \sum_{j=0}^{k-1} \ugn_{s_j} o_j
    =
    \sum_{j=0}^{k-1} \hat\ugn_{s_j} o_j
  \end{equation}
  where $o_j$ is the orientation of edge $s_j$ on the cycle $\mathcal C$.
\end{thm}
\begin{proof}
  Equation~\eqref{eqn:signedcycles} means that
  for every cycle $C$ with incidence vector $x$ we have
  \[
  (\ugn - \hat \ugn)^\tp x = 0
  \]
  Then Lemma~\ref{lem:intdual} implies that $\ugn - \hat \ugn= B^\tp c$
  for some integer vector $c$, and hence $\ugn$ and $\hat \ugn$ are equivalent.
\end{proof}

What this means, in particular, is that in
Example~\ref{ex:rtdirectedcycles} the graph does not have a directed
cycle but it does have a cycle, where edges $1\to 2$ and $2\to 3$ are
oriented in the forward direction, and edge $1 \to 3$ is oriented in
the backward direction. Then $\ugn$ and $\hat \ugn$ are equivalent if
and only if
\[
\ugn_1 + \ugn_2 - \ugn_3 =
\hat\ugn_1 + \hat\ugn_2 - \hat \ugn_3
\]
Since this does not hold for $\lambda$ and $\hat\lambda$ in
that example, those two networks are not equivalent.

One cannot verify equivalence by checking pairs of nodes. That is, it
is not sufficient to simply check the length-2 round trip times, as
the following example shows.

\begin{figure}[htb]
  \hbox to \linewidth{\hbox to 0.5\linewidth{\hss
      \includegraphics[width = 30mm]{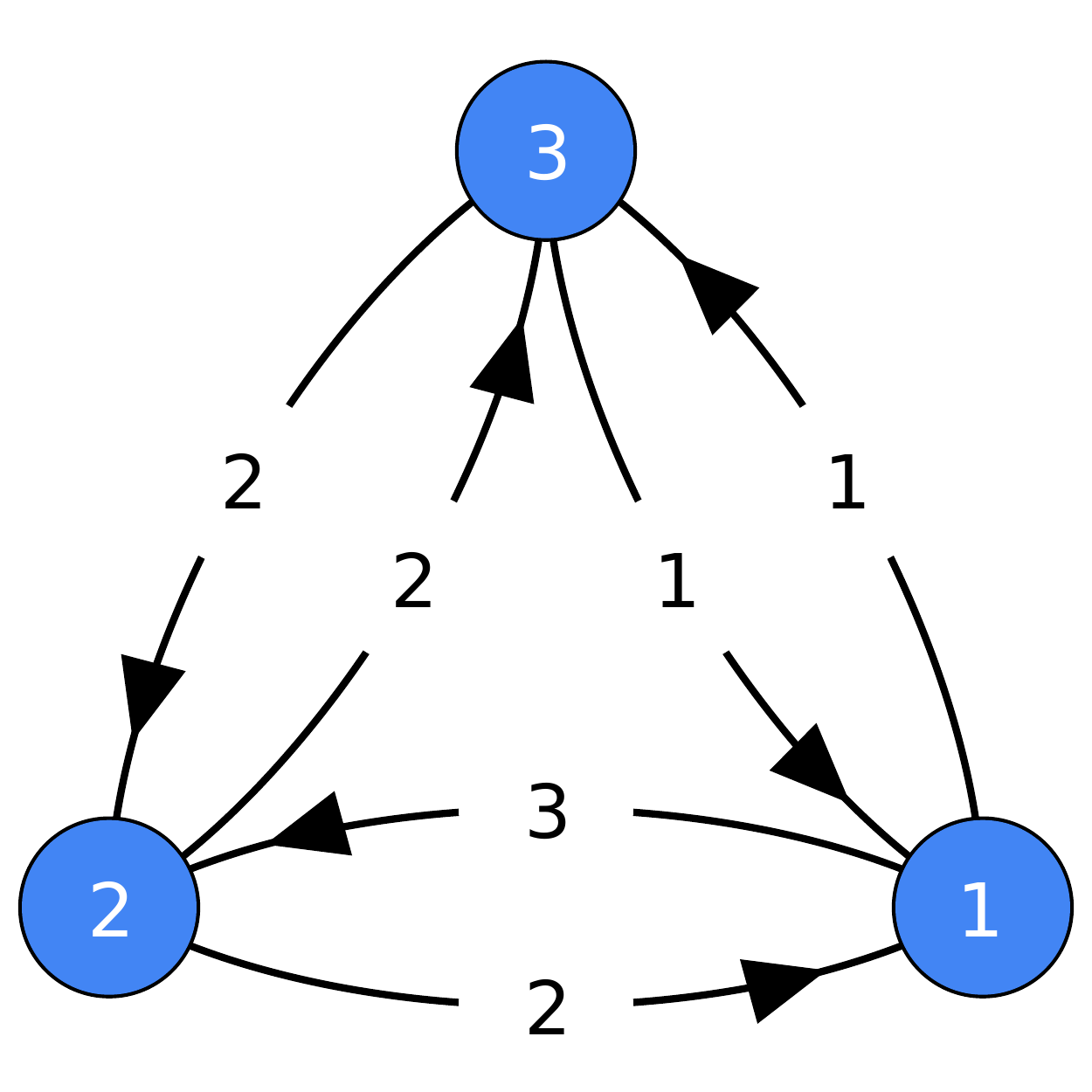}\hss}%
    \hbox to 0.5\linewidth{\hss
      \includegraphics[width = 30mm]{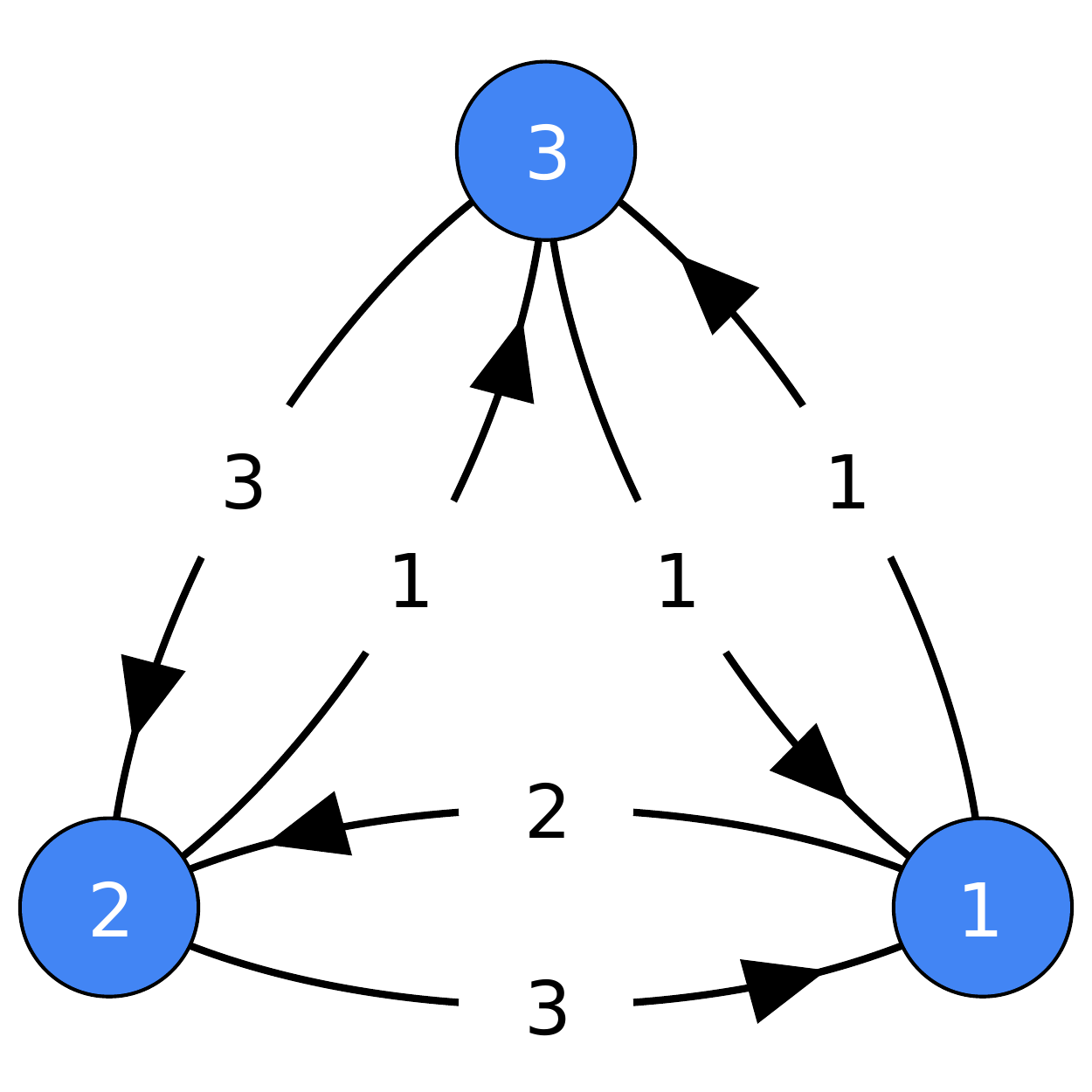}\hss}}
  \hbox to \linewidth{\hss(a)\hskip0.44\linewidth(b)\hss}
  \caption{Logical synchrony networks for Example~\ref{ex:longcycle}}
  \label{fig:lsnex}
\end{figure}

\begin{example}
  \label{ex:longcycle}
  Suppose $\mathcal G$ is the complete
  graph with 3 nodes. For the two logical synchrony networks,
  shown in Figure~\ref{fig:lsnex},
  the length-2 round trip times are
  \begin{align*}
    \ugn_{1 \shortto 2 \shortto 1} &= 5 \\
    \ugn_{2 \shortto 3 \shortto 2} &= 4 \\
    \ugn_{1 \shortto 3 \shortto 1} &= 2
  \end{align*}
  and they are the same for $\hat\ugn$.  However, these networks are not
  equivalent. There is no way to relabel so that the logical
  latencies are the same. This is because the length-3 round trip times
  are $\ugn_{1 \shortto 2 \shortto 3 \shortto 1} = 6$
  and~\hbox{$\hat\ugn_{1 \shortto 2 \shortto 3 \shortto 1} = 4$}.
\end{example}

\paragraph{Invariants.} As shown by the above results, round-trip times
around directed cycles are invariant under relabeling. Cycles
which are not directed also result in invariants which may be
physically measured and interpreted. We give some examples below.

\begin{example}
  \label{ex:invtri}
  Figure~\ref{fig:invtri} shows a triangle graph in which node 1 sends
  frame $A$ to node 3, and simultaneously sends frame $B$ to node 3
  via node 2. Then $\receive(B) - \receive(A)$ is measured in
  localticks at node 3, and it is invariant under relabeling.
\end{example}

\begin{figure}[htb]
  \hbox to \linewidth{\hss\begin{overpic}[width=30mm]{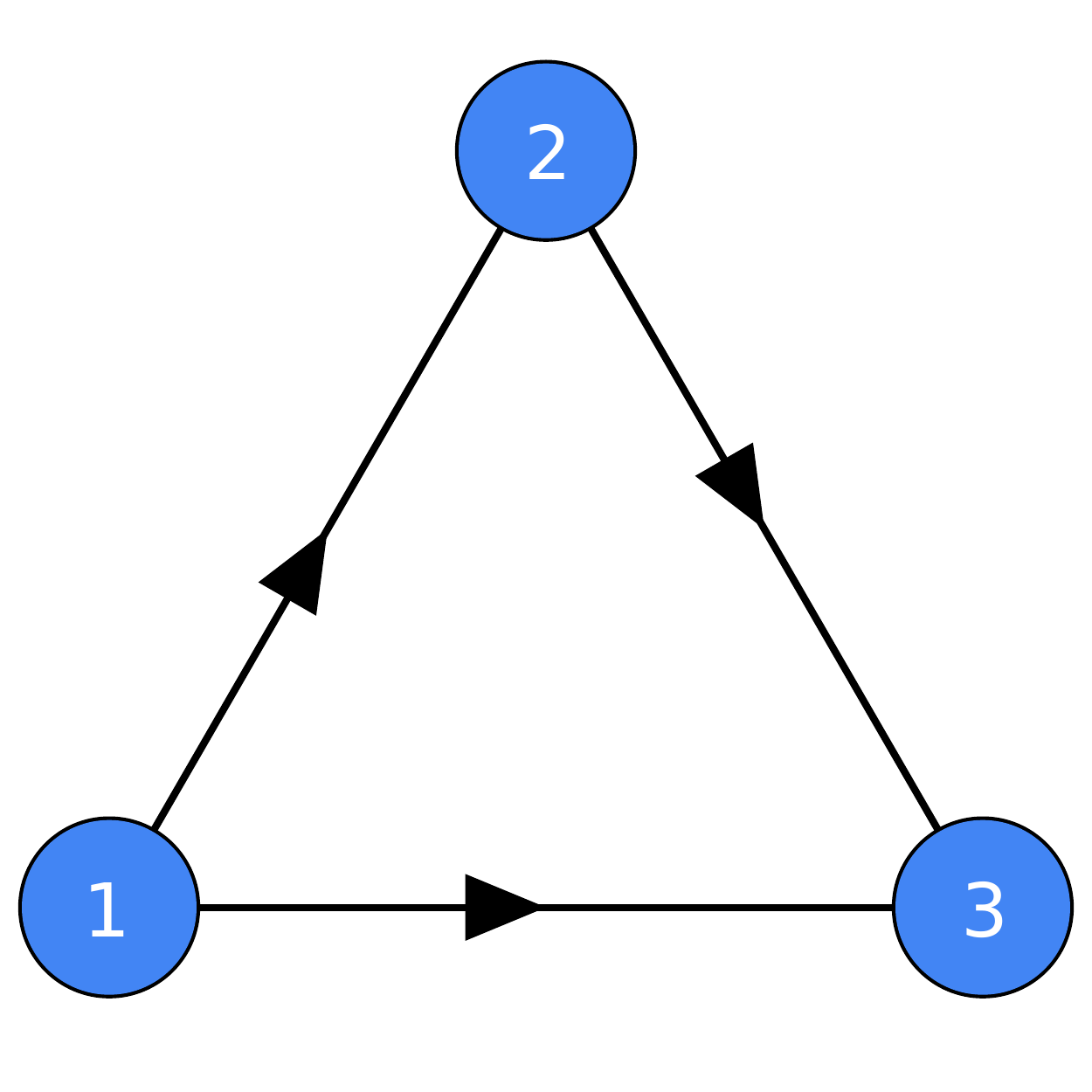}
          \put(45,7){$A$}
          \put(17,52){$B$}
          \put(74,52){$B$}
    \end{overpic}\hss}
  \caption{Triangle invariant}
  \label{fig:invtri}
\end{figure}

\begin{example}
  \label{ex:invdia}
  Figure~\ref{fig:invdia} shows a square graph. Here node 1 sends
  frame $A$ to node 2 and simultaneously sends frame $B$ to node
  4. Node 3 sends frame $C$ to node 2 and simultaneously sends frame
  $D$ to node 4.  Note that the transmissions of node 1 and node 3 are
  not synchronized with each other.  Then the quantity
  \[
  (\receive(A) - \receive(C)) - (\receive(B) - \receive(D))
  \]
  is invariant under clock relabelings.
\end{example}
\begin{figure}[htb]
  \hbox to \linewidth{\hss\begin{overpic}[width=30mm]{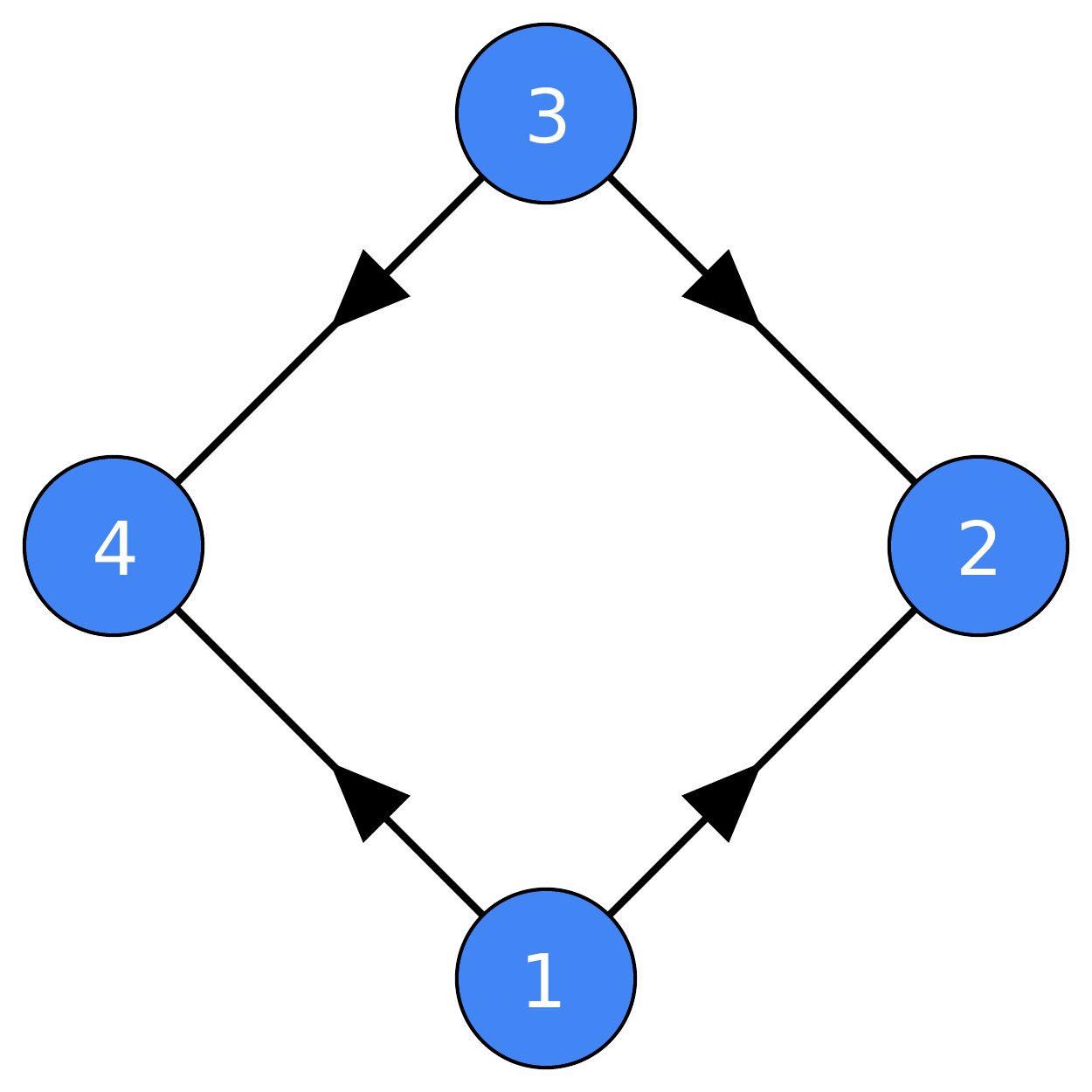}
          \put(70,20){$A$}
          \put(20,20){$B$}
          \put(72,72){$C$}
          \put(20,72){$D$}
    \end{overpic}\hss}
  \caption{Diamond invariant}
  \label{fig:invdia}
\end{figure}
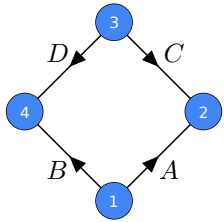

Equivalent networks can have different logical latencies, but must
have the same round-trip times. The question of how much freedom this
leaves is interesting, and has an important consequence which we
discuss below. We first show that one can set the logical latencies
arbitrarily on any spanning tree.

\begin{thm}
  \label{thm:labtree}
  Suppose $\mathcal G, \ugn$  is  a logical synchrony network,
  where $\mathcal G= (\mathcal V, \mathcal E)$. Suppose $\tree
  \subset \mathcal E$ is a spanning tree. Then for any
  $\gamma:  \tree \to \Z$ there exists $c\in\Z^n$ such that
  \[
  \gamma_\itoj = \ugn_\itoj + c_j - c_i \text{ for all } i \to j \in \tree
  \]
\end{thm}
\begin{proof}
  We would like to show that there exists $c\in\Z^{n}$ such that
  \[
  \bmat{I & 0 } (\lambda - \gamma) = \bmat{I & 0} B^\tp c
  \]
  Let $y_1$ be the left-hand side, then using Theorem~\ref{thm:smith},
  this is equivalent to
  \[
  y_1 = \bmat{B_{11}^\tp & -B_{11}\one} c
  \]
  and hence we may choose
  \[
  c = \bmat{B_{11}^{-\tp} y_1 \\ 0}
  \]
  which is integral since $B_{11}$ is unimodular.
\end{proof}

We can use this result in the following way. There is no requirement
within this framework that logical latencies be nonnegative. However, it
turns out that any logical synchrony network which has nonnegative
round-trip times is equivalent to one with nonnegative logical
latencies.  We state and prove this result below. This result will be
useful when we discuss multiclock networks in the subsequent section.

\begin{theorem}
  \label{thm:shortest}
  Suppose $\mathcal G, \ugn$ is a logical synchrony network with
  $\mathcal G$ strongly connected, and for every directed cycle
  $\mathcal C$ the round-trip logical latency $\lambda_\mathcal C$ is
  nonnegative.  Then there exists an equivalent LSN with edge weights
  $\hat \lambda$ which are nonnegative.
\end{theorem}
\begin{proof}
  Pick a node $r$. Since the graph has no negative cycles, there
  exists a spanning tree $\tree$, rooted at $r$, with edges
  directed away from the root, each of whose paths is a shortest
  path~\cite{tarjan}. Use Theorem~\ref{thm:labtree} to construct $c$
  such that
  \[
  \ugn_\itoj + c_j - c_i = 0 \text{ for all } i \to j \in \tree
  \]
  As a result, we have $\ugn_\itoj = c_i - c_j$ for all edges $i \to
  j$ in the tree $\tree$. Denote by $t_{i\shortto k}$ the length of the
  path from $i$ to $k$ in the tree. Then we have $t_{i\shortto k} = c_i - c_k$.

  Since this is a shortest path tree, we have for any edge $i \to j$
  \[
  t_{r \shortto i} + \lambda_\itoj \geq t_{r \shortto j}
  \]
  because the path in the tree from $r$ to $j$ must be no longer than
  the path via node $i$. Therefore
  \[
  c_r - c_i + \lambda_\itoj \geq c_r - c_j
  \]
  Setting $\hat \lambda_\itoj = \lambda_\itoj +c_j - c_i$ for all edges
  we find $\hat \lambda_\itoj \geq 0$ as desired.
\end{proof}

This result says that, if we have a shortest path tree, we can relabel
the clocks so that the logical latency is zero on all edges of that
tree, and with that new labeling the logical latency will be
nonnegative on every tree edge. An example is given in
Figure~\ref{fig:lsnneg}.

\begin{figure}[htb]
  \hbox to \linewidth{\hss
    \includegraphics[width = 60mm]{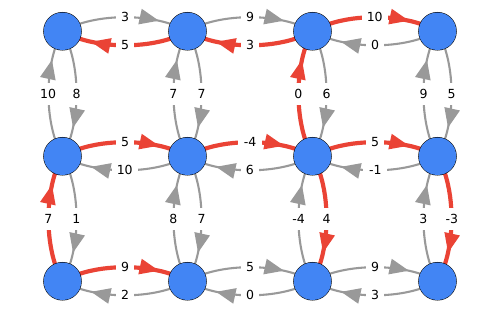}\hss}
  \par\bigskip
  \hbox to \linewidth{\hss
      \includegraphics[width = 60mm]{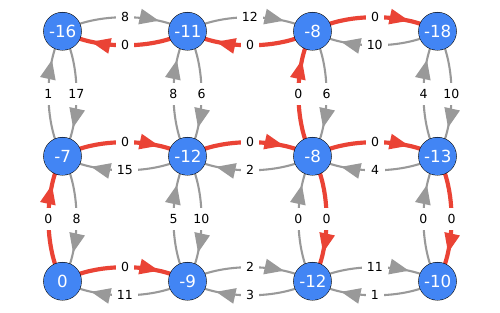}\hss}
  \caption{Relabeling so that logical latencies are nonnegative. The
    upper graph shows edges labeled with $\lambda$. The root node is
    in the lower left, and the
    shortest-path spanning tree is shown in red. The lower graph
    shows an equivalent LSN, with nodes $i$ labeled with $c_i$, and
    the corresponding logical latencies $\hat \lambda_\itoj =
    \lambda_\itoj +c_j - c_i$. All logical latencies in this graph are
    nonnegative. }
  \label{fig:lsnneg}
\end{figure}

Note also that an edge having zero logical latency does not imply that
communication between the endpoints is instantaneous; only that the
numerical value of the time at which the frame is received is equal to
the numerical value of the time at which it was sent.

\section{Multiclock networks}
\label{sec:mc}

The objective of this section is to build a model of a
physical system, and relate its correctness to the invariants of the
previous section.  We introduce a model in which there are physical
clocks at each node, and the nodes pass data to each other, according
to specific timed sequential communications which occur through
FIFOs. We call such a system a \emph{multiclock network}.  We show
that the latencies that arise satisfy exactly the semantics of the
abstract latencies of logical synchrony networks. We further show that
the natural requirements that the FIFO occupancies be bounded leads to
the physical requirement that round trip times are nonnegative.  In
other words, building a correct multiclock network will result in a
correct logical synchrony network.

We formulate the relationship between events on a
network in terms of physical clocks, leading to a mathematical
definition called the multiclock network. We show that
multiclock networks are special types of logical synchrony networks.

We will use $t$ to denote an idealized notion of time, called
\emph{wall-clock time}, or \emph{ideal time}~\cite{andre}.  Time on
the network is \emph{multiform}~\cite{berry}, in the sense that the
nodes on the network each maintain their own sense of time.  At each
node, there is a real-valued clock, denoted by $\theta_i$.  Its units are the
\emph{localticks}. We refer to the value $\theta_i$ as the \emph{local
time} or \emph{phase} at node $i$. Local time has no quantitative
relationship to physical or wall-clock time.  In particular, we do not
view $\theta_i$ as an approximation to wall-clock time and
consequently clocks at two distinct nodes are inherently unrelated.

At a node $i$, a processor can read the value $\theta_i$, its own
clock, but cannot access the value $\theta_j$ at any other node $j\neq
i$.  We mathematically model $\theta_i$ as a function of physical
time~$t$, so that $\theta_i:\R \to \R$, without implying anything
about its construction; it simply means that if at physical time $t$ a
hypothetical outside observer were to read clock $i$, it would read
value $\theta_i(t)$. What is required is that $\theta_i$ is continuous
and increasing, so that $\theta_i(s) < \theta_i(t)$ if $s < t$. We
emphasize again that this does not imply that any processes running on
the system can access wall-clock time $t$. The quantity $\theta_i$ is
not related to physical time.

At times $t$ where $\theta_i$ is differentiable, we define the
frequency~$\omega_i$ of the clock $\theta_i$ by
\[
\omega_i(t) = \frac{d \theta_i(t)}{dt}
\]
At a node $i$, a clock generates an infinite sequence of events, also
referred to as \emph{localticks}, which happen whenever $\theta_i$ is
an integer.  Clocks are not required to be periodic, and
this definition of frequency is applicable in the general aperiodic case.
Clocks at different nodes may have very different frequencies.
 If the frequency at node $i$ is large, then events at that
node occur more often.

We model the process of frame transmission from node $i$ to node $j$
as a FIFO, but real-world implementations are likely to consist of
uninterrupted physical communication streams feeding into memory
buffers.  Every node can access the output (or head) of the FIFO
corresponding to each of its incoming links, and the input (or tail)
of the FIFO corresponding to each of its outbound links. We will
discuss below the requirement that FIFOs neither overflow nor
underflow.

\paragraph{Logical synchrony in multiclock networks.}
With every localtick, node $i$ inserts a frame at the tail of each of
its outgoing link FIFOs and removes a frame from the head of each of
its incoming link FIFOs.  This lock-step alignment of input and output
is the fundamental synchronization mechanism that imposes logical
synchrony upon the network. At each node, with every localtick, one
frame is removed from each incoming FIFO and one frame is sent on each
outgoing FIFO.

\paragraph{Formal definition of multiclock network.}
We now turn to a mathematical model that will enable us to analyze the
behavior of this system.

\begin{defn}
  \label{defn:multiclock}
  A \eemph{multiclock network} is a directed graph $\mathcal G =
  (\nodes,\edges)$ together with continuous increasing functions
  ${\theta_i:\R \to \R}$ for each $i\in\nodes$, and edge weights
  $\ugn:\edges \to \Z$.
\end{defn}

This definition contains the entire evolution of the clock phases
$\theta_i$, and the link properties $\lambda_\itoj$. We will discuss
the physical meaning of $\lambda_\itoj$ below. Unlike the logical
synchrony network, where events are abstract and have no physical time
associated with them, in a multiclock network the global timing of all
events is defined by the clocks $\theta$. We will show that a
multiclock network is a special case of a logical synchrony network,
and the constants $\lambda$ are the associated logical latencies. To
do this, we  model the behavior of the FIFOs connecting the nodes.

\paragraph{FIFO model.} If $i \to j$ in the graph $\mathcal G$, then
there is a FIFO connecting node $i$ to node $j$.  With every localtick
at node $i$, a frame is added to this FIFO, and with every localtick
at node $j$, a frame is removed from the FIFO. We number the frames in
each FIFO by $k\in\Z$,  according to the localtick at the sender,
and the frames in the FIFO are those with $k$ satisfying
\[
\alpha_\itoj(t) \leq k \leq \beta_\itoj(t)
\]
where $\alpha$ and $\beta$ specify which frames are currently in the
FIFO at time $t$. The FIFO model is as follows.
\begin{align}
  \label{eqn:fifobeta}
  \beta_\itoj(t) & = \floor{\theta_i(t)}  \\
  \label{eqn:fifoalpha}
  \alpha_\itoj(t) & = \floor{\theta_j(t)} - \lambda_\itoj + 1
\end{align}
Equation~\eqref{eqn:fifobeta} means that frames are added with each
localtick at the sender, and numbered according to the sender's clock.
Equation~\eqref{eqn:fifoalpha} means that frames are removed with each
localtick at the receiver. The constant $\lambda$ is to account for
the offset between the frame numbers in the FIFO and the clock labels
at the receiver. (We add $1$ for convenience.) This offset must be
constant, since one frame is removed for each receiver localtick.
This constant is specified by the multiclock network model in
Definition~\ref{defn:multiclock}.

This model precisely specifies the location of every frame on the
network at all times $t$. In particular, this determines the FIFO
occupancy at startup. For any time $t_0$, the specification of
$\lambda$ is equivalent to specifying the occupancy of the FIFOs at
time $t_0$. This allows us to have a well-defined FIFO occupancy
without requiring an explicit model of startup.

\paragraph{Logical latency.}

Logical latency is the fundamental quantity which characterizes
the discrete behavior of a network, and allows us to ignore the
details of the clocks~$\theta_i$. The idea is that we can understand
the logical structure of the network, such as the events, the
execution model, and causality, without needing to know specific
wall-clock times at which these things occur.

We now show that the quantity $\lambda_\itoj$ corresponds to the
logical latency.  Suppose a frame is sent from node $i$ at localtick
$k\in\Z$, and wall-clock time $\tsend^k$. Then $\theta_i(\tsend^k) = k$.
Let the time which it is received at node $j$ be denoted by $\trec^k$.
Both $\tsend^k$ and $\trec^k$ are wall-clock times, and apart from the
causality constraint that the frame must be received after it is sent,
there is no constraint on the difference between these times; that is,
the \emph{physical latency} $\trec^k - \tsend^k$ may be large or small.
In general, physical latency will be affected by both the number of
frames in the FIFO $i \to j$ as well as the time required for a frame
to be physically transmitted. We do not presuppose requirements on the
physical latency.

\begin{lem}
  Suppose frame $k$ is sent from node $i$ to node~$j$. Then
  $\tsend^k$ and $\trec^k$ satisfy
  \begin{align}
    \label{eqn:tsend}
    \theta_i(\tsend^k) & = k \\
    \label{eqn:trec}
    \theta_j(\trec^k) &= k + \lambda_\itoj
  \end{align}
  and hence the logical
  latency is given by
  \begin{equation}
    \label{eqn:loglat}
    \begin{aligned}
    \ugn_\itoj &= \theta_j(\trec^k) - \theta_i(\tsend^k)
    \end{aligned}
  \end{equation}
\end{lem}
\begin{proof}
  Since frames in the FIFO $i \to j$  are numbered according
  to the sender's clock, we have
  \[
  \tsend^k = \inf\{ t \mid \beta_\itoj(t) = k \}
  \]
  that is, $\tsend^k$ is the earliest time at which frame $k$ is in the
  FIFO from $i$ to $j$. Since the floor function is right continuous,
  this gives equation~\eqref{eqn:tsend}.
  Similarly, we have
  \[
  \trec^k = \inf \{ t \mid \alpha_\itoj(t) = k + 1\}
  \]
  which is the first time $t$ at which the lowest-numbered frame
  in the FIFO is number $k+1$, and therefore this is the time at which
  frame $k$ has just left the FIFO, and hence has just arrived at the
  destination.
  This implies equation~\eqref{eqn:trec}, and the logical latency follows.
\end{proof}
Unlike the physical latency $\trec - \tsend$, the logical latency
$\theta_j(\trec^k) - \theta_i(\tsend^k)$ does not change over
time. Note also that the logical latency is an integer.  Since the
logical latency is constant, we can conclude that every multiclock
network is a logical synchrony network; more precisely, the logical latencies
defined by the multiclock network satisfy the same properties as those
of a logical synchrony network.

\subsection{Realizability}

We now turn to an analysis of the occupancy of the {FIFO}s in more
detail.  A frame is considered \emph{in-transit} from $i \to j$ at
time $t$ if it has been sent by node $i$ but not yet received by node
$j$; that is, if it is in the FIFO from $i$ to $j$.  Define
$\nu_\itoj(t)$ to be the number of frames in transit $i \to j$.
Then we have
\begin{align}
  \nonumber
  \nu_\itoj(t)
  & = \beta_\itoj(t) - \alpha_\itoj(t) + 1 \\
  \label{eqn:ugn}
  &= \floor{\theta_i(t)} - \floor{\theta_j(t)} + \ugn_\itoj
\end{align}
and this holds for all $t$. Here we can see that
the constant $\lambda_\itoj$ is a
property of the link $i \to j$, which determines the relationship
between the clock phases at each end of the link and the number of
frames in transit.

So far in this model, there is nothing that prevents the FIFO
occupancy on an edge $i\to j$ from becoming negative.  If the clock at
node $\theta_j$ has a higher frequency than the clock at $\theta_i$,
and if that frequency difference is maintained for long enough, then
the FIFO $i \to j$ will be rapidly emptied.  In this case, $\theta_j$
will become much larger than $\theta_i$, and from~\eqref{eqn:ugn} we
have that $\nu_\itoj$ will become negative. Similarly, the FIFO will
overflow if the frequencies become imbalanced in the other direction.
In \cite{reframing} a technique using a dynamically switching control
algorithm is presented that allows prevention of such behaviors.  We
make the following definition.

\begin{defn}
  \label{defn:realizable}
  A multiclock network is called \eemph{realizable} if
  there exists $\nu_\text{max} \in \R$ such that
  for all edges $i \to j$
 \begin{equation}
   \label{eqn:deviation1}
   0 \leq  \nu_\itoj(t)  \leq \nu_\text{max} \quad \text{for all }t \in \R
 \end{equation}

\end{defn}
Note that this requirement must hold for all positive and negative time $t$.
The terminology here is chosen to be suggestive, in that we would like
a condition which implies that we can physically implement a
multiclock network. A physically necessary condition is that the FIFO
occupancies are bounded and cannot be negative.

\paragraph{Cycles and conservation of frames.}
Cycles within a multiclock network have several important
properties. The first is \emph{conservation of frames}, as follows.
\begin{thm}
  \label{thm:cycles}
  Suppose $\mathcal C =  v_0,s_0,v_1,s_1,\dots,s_{k-1},v_k$
 is a directed cycle in a multiclock network. Then
  \[
  \sum_{i=0}^{k-1} \nu_{s_i}(t) =
  \lambda_{\mathcal C}
  \]
  In particular, the number of frames in transit around the cycle is
  constant, and is the sum of the logical latencies on the cycle.
\end{thm}
\begin{proof}
  The proof follows immediately from~\eqref{eqn:ugn}.
\end{proof}

An immediate corollary of this is that, in a physical network, if
every edge of $\mathcal{G}$ is on a cycle, then the number of frames
in the network is finite and the upper bound condition for
realizability is satisfied. This is the case, for example, in a
strongly connected graph. Note that this holds because, in a
physical network, the FIFO occupancy cannot be negative. It is not the
case that the FIFO model used here implies that $\nu$ is upper
bounded, since in the model some FIFO lengths may become large and
negative while others become large and positive.

This theorem is particularly evocative in the simple and common case
where we have two nodes $i$, $j$ connected by links in both
directions. In this case, whenever $i$ receives a frame, it removes it
from it's incoming FIFO from $j$, and adds a new frame to the outgoing
FIFO to~$j$.  Thus the sum of the occupancies of the two FIFOs is
constant.

The following result relates round trip times to realizability.

\begin{thm}
  Suppose $\mathcal{C}$ is a cycle in a realizable multiclock
  network. Then $\lambda_\mathcal{C} \geq  0$.
\end{thm}
\begin{proof}
  This follows immediately from Theorem~\ref{thm:cycles} and
  Definition~\ref{defn:realizable}.
\end{proof}

That is, a realizable multiclock network has the important physical
property that all round-trip times are nonnegative.  The
monotonic property of $\theta$ implies that this holds in both
localticks and wall-clock time. No matter what path a frame takes
around the network, it cannot arrive back at its starting point before
it was sent.  However, it is possible, within the class of realizable
networks defined so far, for this sum to be equal to zero. In this
case one would have a frame arrive at the time it is sent.  This would
require some pathological conditions on the clocks.
This is an extreme case corresponding to the limit where
frames spend zero time in the FIFOs, which in a physical network would
require that the link have zero link latency.  For example, in the
case of a length 2 cycle between nodes $i$ and $j$, we would need
$\theta_j(t) = \theta_i(t) + \lambda_\itoj$ and $\theta_i(t)
= \theta_j(t) + \lambda_{\jtoi}$, which would give $\lambda_\itoj
= \lambda_\jtoi$. Since the clocks are related by integer constants,
they tick at exactly the same times.

\subsection{Equivalent synchronous systems}
\label{subsec:ess}
We now consider the class of perfectly synchronous systems, where all
of the nodes of the graph share a single clock. The links between the
nodes are FIFOs as before, and as a result of the synchronous
assumption their occupancies are constant.  This is a particular
instance of the multiclock network where all clocks $\theta_i$ are
equal.

Such a system has an extended graph, and it has logical latencies
which do not change with time, and are equal to the occupancies of the
FIFOs, according to~\eqref{eqn:ugn}. Because the system is
synchronous, the FIFOs behave like a chain of delay buffers.
The corresponding execution model, defined by the extended graph,
is identical to that of a logical synchrony network with the same
logical latencies. Said another way, a logical synchrony network
is equivalent to a perfectly synchronous network of processors connected by
delay buffers with occupancies given by the logical latencies.

This suggests the following question; what happens if we have a logical
synchrony network where one or more of the edges has a negative
logical latency?  Using Theorem~\ref{thm:shortest}, we know that if a
network has nonnegative round-trip times, one can relabel the clocks
so that all logical latencies are nonnegative. Hence any physically
constructible multiclock network is equivalent to a perfectly
synchronous network.

\section{The bittide mechanism}

We now turn to a specific form of multiclock network
which can be implemented on modern networking hardware.
In
Section~\ref{sec:mc} we have already discussed one of the key
components of this, specifically that with each localtick, a node
removes one frame from the head of every incoming FIFO, and sends one
frame on every outgoing FIFO.  However, this is not enough for
implementation, since we must ensure that the occupancies of the FIFOs
neither underflow nor overflow.

In the \bittide model, the FIFO connecting node $i$ to node $j$ is
composed of two parts, connected sequentially.  The first part is a
communication link, which has a latency $l_\itoj$, the number of
\emph{wall-clock} seconds it takes to send a frame across the
link. The second part is called the \emph{elastic buffer}. It is a
FIFO which is located at the destination node $j$. Node $i$ sends
frames, via the communication link to node $j$, where they are
inserted at the tail end of the elastic buffer. We assume that the
communication link cannot reorder frames, and so together the
communication link and the elastic buffer behave as a single FIFO.

Each node has an elastic buffer for each of its incoming links.  With
each clock localtick, it does two things; first, it removes a frame from the
head of each of the elastic buffers and passes that frame to the
processor core; second, the core sends one frame on each outgoing
communication link.

The purpose of this structure is as follows.  An implementation
of \bittide has nodes whose hardware oscillators are
adjustable. The elastic buffer occupancies provide
information regarding the relative clock frequencies of the node
compared to its incoming neighbors.  This allows the oscillators to be
adjusted in real-time, by each node, based on measurements of the
occupancy of the elastic buffers.  Off-the-shelf modules are available
which provide fine-grained control of the oscillator frequency.
Specifically, if we have an edge $i \to j$, and node $i$
has a lower clock frequency that node $j$, then the corresponding
elastic buffer at node $j$ will start to drain. Conversely, if node
$i$ has a higher clock frequency, the elastic buffer will start to
fill. Node $j$ can therefore use the occupancy of the elastic buffers
to adjust its own clock frequency. If, on average, it's buffers are
falling below half-full, the node can reduce its clock frequency, and
conversely.

\begin{figure}[htb]
  \hbox to \linewidth{\hss\begin{overpic}[width=35mm]{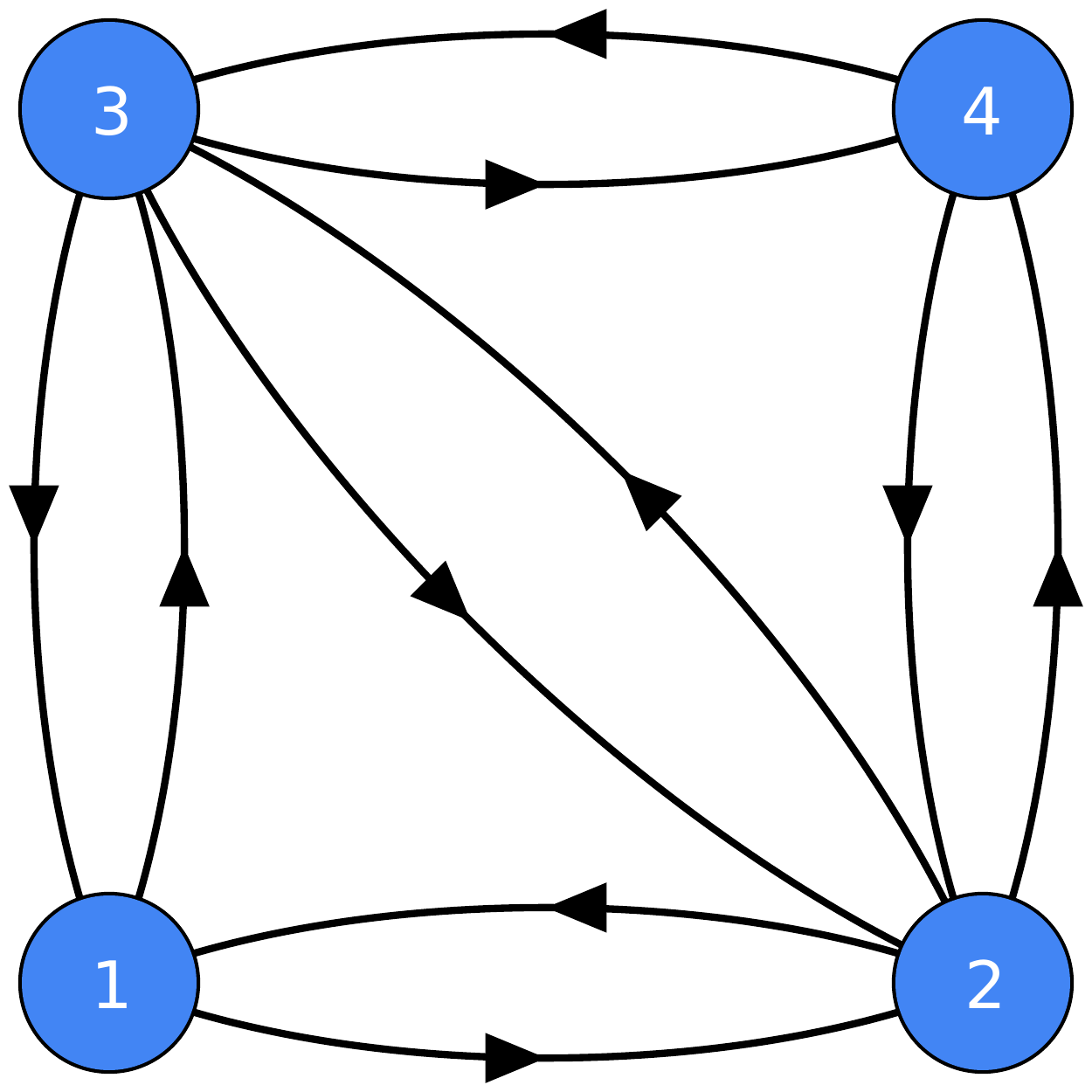}
    \end{overpic}\hss}
  \caption{Graph for \bittide simulation}
  \label{fig:btgraph}
\end{figure}

This mechanism was originally proposed in~\cite{spalink_2006}. Further
refinements to the implementation were developed
in~\cite{bms,res,reframing}.  These papers show that, provided the
frequency corrections are chosen appropriately, this mechanism will
ensure that elastic buffers never underflow or overflow.
A simple mechanism for doing this is to control
the \emph{correction}. Adjustable oscillators allow
choosing a value for correction $c$, which causes the frequency
$\omega$ to become
\[
\omega = (1 + \alpha c) \wu
\]      
Here $\wu$ is the base frequency of the oscillator, which is only
known approximately, and $\alpha$ is small, of the order of $10^{-6}$.
Let $\beta_{\itoj}$ be the occupancy of the elastic buffer at node $j$
for the link from $i$ to $j$.  Each node $j$ polls the hardware to
observe these quantities, and sets the correction at node $j$ to be
\[
c = k_p \sum_{i \mid i \to j} (\beta_{\itoj} - \beta_0)
\]
where $k_p$ is a positive constant, and the sum is over all links
which are incoming to $j$. The value $\beta_0$ is a fixed offset.  For
an appropriate choice of $k_p$, all of the the frequencies converge to
the same steady-state value.  See~\cite{bms,res,reframing} for more
details.

An example simulation of the clock dynamics is
in Figure~\ref{fig:btgraph}. The time
evolution of the clock frequency $\omega$ and the buffer occupancy
$\beta$ is shown in Figure~\ref{fig:btsim}, with the buffer occupancy
for edge $i\to j$ labeled $i,j$. For this simulation, the link
latencies $l_\itoj = 1\text{ns}$. Note that in this simulation the
parameters are chosen so that
the dynamics of the system are clearly visible.
In particular, the nodes start
at frequencies $1.1, 1.4, 1.8, 2.0$ GHz, and in practical hardware
systems typically the frequencies at startup would be separated by
less than one part in $10^5$.  Similarly, the control algorithm
parameters are set so that convergence is slow and the equilibrium
buffer occupancies are large, between 25 and 75 frames, whereas
in practice (\eg, in the hardware of~\cite{qbay}) these parameters
are chosen to keep the buffers much smaller. With more realistic parameters
the dynamics follow the same general pattern, but are less visible on a plot. 

\paragraph{Available implementations.}
There are three open-source efforts addressing \bittide and logically
synchronous systems. The first is the hardware description, written in
Clash, available at~\cite{qbay}. This may be compiled onto
standard FPGA boards, linked to controlled oscillator boards.  Second,
there is a simulator called \emph{Callisto}~\cite{callisto}, which is
written in Julia, and simulates the dynamics of the oscillators and
the occupancies of the elastic buffers.  Finally, there is the
\emph{Aegir} simulator~\cite{aegir}, written in Rust, which is a
functional simulation of a logical synchrony network.

\begin{figure}[t]
  \centerline{\begin{overpic}[width=1\linewidth]{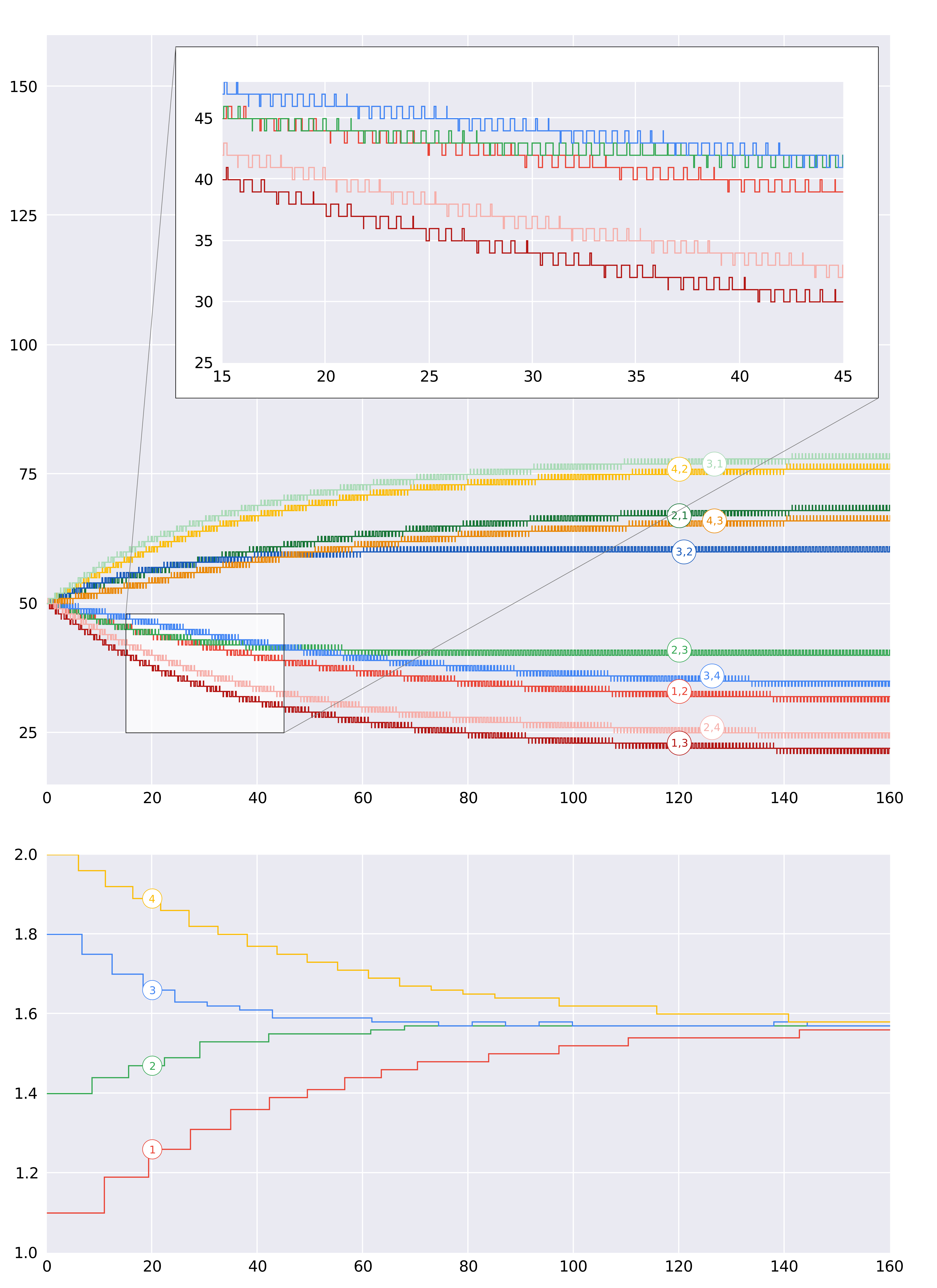}
      \put(3,17.5){\llap{\scalebox{0.6}{$\omega(GHz)$}}}
      \put(0,68){\llap{\small$\beta$}}
      \put(38,-1){\clap{\small$t (ns)$}}
  \end{overpic}}
  \vspace*{3mm}
  \caption{Occupancy and frequency of the \bittide system.}
  \label{fig:btsim}
\end{figure}

\section{Related work}

The seminal work of Lamport~\cite{lamport} presents a formal framework
for clocks in distributed systems, which in particular defined an
ordering on a directed graph corresponding to temporal relationships
between events, and a global scalar clock which was consistent with
that ordering.  Subsequent work~\cite{mattern,fidge} developed the
notion of vector clocks, where each node in a network maintains a
vector notion of time which captures exactly the ordering defined by
the graph. The synchronization mechanism of \bittide was first
proposed in~\cite{spalink_2006}. Subsequent works include~\cite{bms},
which developed a mathematical model of the synchronization layer,
and~\cite{res}, which analyzed its performance properties.

Ever since the first distributed systems, synchronous execution has
been a gold standard for formal reasoning, provable
correctness properties, and ability to express efficient
algorithms~\cite{lamport_reliable_distr_systems_1978,
  dwork_consensus_1988, liskov_practical_1991,le2003polychrony}.
As a consequence, the
domain of synchronous execution has been studied extensively, in
particular in the context of cyber-physical systems. Cyber-physical
systems interact with physical processes, and
Lee~\cite{lee:time:cacm:2009} argues that integrating the notion of
time in system architecture, programming languages and software
components leads to the development of predictable and repeatable
systems.

Reasoning about distributed systems has led to the definition of both
execution models and parallel programming models. Kahn Process
Networks~\cite{kahn1974} is one of the most general; while it does not
involve time or synchronization explicitly, processes in a Kahn
process network communicate through blocking FIFOs, and thus
synchronize implicitly through the communication queues. An important
distinction between \bittide and the Kahn Process Networks is that the
former does not make use of blocking.

Synchrony, and its most common representation as a global time
reference, led to the definition of multiple models of
computation. For example, Synchronous Dataflow~\cite{lee:sdf:1987}
enables static scheduling of tasks to resources; Timed Concurrent
Sequential Processes (Timed CSP)~\cite{reed1988metric} develop a model
of real-time execution in concurrent systems; Globally Asynchronous,
Locally Synchronous (GALS) communication
models~\cite{potop-butucaru_concurrency_2004} address the issue of
mapping a synchronous specification to existing systems which are
asynchronous.

Henzinger et al.~\cite{henzinger_embedded_2001} introduce the concept
of \emph{logical execution} and Kopetz et
al.~\cite{kopetz_time-triggered_2003} introduce Time-Triggered
Architectures (TTAs) as a system architecture where time is a
first-order quantity and they take advantage of the global time
reference to exploit some of the desirable properties of synchronous
execution: precisely defined interfaces, simpler communication and
agreement protocols, and timeliness guarantees.

Synchronous programming models led to synchronous programming languages, e.g.,
Esterel~\cite{berry:esterel:1998}, Lustre~\cite{halbwachs:lustre:1991},
Signal~\cite{benveniste:signal:1991}, and the development of tools to formally
analyze their execution correctness as well as compilers to generate correct
synchronizing code for embedded~\cite{benveniste:12yearsofsynchrony:2003} or
multicore platforms~\cite{didier_sheep_2019}. This created a virtuous cycle --
as researchers understood better properties and embedded them into languages
and tools, they drove the adoption of synchronous execution and formal tools
for a number of industrial control applications, avionics, and critical system
components.

\section{Conclusions}

This paper has presented the logical synchrony framework. We have
shown how this may be used to enable processes on a network of
distributed machines to coordinate as if they were synchronized, even
if the the clocks on the individual cores are only imperfectly
synchronized. We have discussed the \bittide mechanism for
implementing logical synchrony, how it is abstracted as a multiclock network,
and how that corresponds to a further abstraction called the logical
synchrony network (LSN). We have analyzed the invariant properties
of these networks, and shown how these clocks provide predictable
logical latencies on the network.

\section{Acknowledgments}

The ideas for this paper came about through much collaboration. In
particular, we would like to thank Nathan Allen, Pouya Dormiani, Chase
Hensel, Logan Kenwright, Robert O'Callahan, Chris Pearce, Dumitru
Potop-Butucaru, and Partha Roop for many stimulating discussions about
this work. Robert had the idea for the proof of
Theorem~\ref{thm:shortest}.

\end{document}